%% file: main.tex
\documentclass[11pt,letterpaper]{article}
\usepackage[margin=1in]{geometry}

\usepackage{graphicx} 

\usepackage{amsmath}
\usepackage{amssymb, amsthm, amsfonts, graphicx, color, subcaption, enumerate, bm, enumitem, array, mathtools}
\usepackage{mathdots}
\usepackage{tikz}
\usetikzlibrary{arrows.meta}
\usepackage{booktabs}
\usepackage{xspace}

\definecolor{purple}{HTML}{1A1AB3}
\definecolor{red}{HTML}{D32F2F}
\definecolor{green}{HTML}{2E7D32}

\title{Local Certification of Vertex and Edge Connectivity}

\author{Yi-Jun Chang\footnote{National University of Singapore. ORCID: 0000-0002-0109-2432. Email: cyijun@nus.edu.sg}  \and Yi-Xuan Lee\footnote{National University of Singapore. ORCID: 0009-0002-9754-9931. Email: r12921065@ntu.edu.tw}  \and Meng-Tsung Tsai\footnote{Academia Sinica. ORCID: 0000-0002-2243-8666. Email: mttsai@as.edu.tw}}

\date{}

\usepackage{thmtools, thm-restate}

\usepackage{multirow}

\usepackage[style=trad-alpha,natbib=true,maxcitenames=100]{biblatex}
\usepackage{algpseudocode,algorithm}

\newcommand{\dist}{\operatorname{dist}}
\newcommand{\depth}{\operatorname{depth}}

\newcommand{\diam}{{D}}

\newcommand{\ID}{\operatorname{ID}}
\newcommand{\seg}{\operatorname{seg}}
\newcommand{\rank}{\operatorname{rank}}

\newcommand{\Long}{\textsf{Long}}
\newcommand{\Short}{\textsf{Short}}

\newcommand{\DISJ}{\textsf{DISJ}}
\newcommand{\true}{\text{\small\texttt{TRUE}}}
\newcommand{\false}{\text{\small\texttt{FALSE}}}

\newcommand{\poly}{\operatorname{poly}}

\newcommand{\Fcal}{\mathcal{F}}

\newcommand{\Gcal}{\mathcal{G}}

\newcommand{\Pcal}{\mathcal{P}}

\newcommand{\reject}{\textsf{reject}\xspace}
\newcommand{\accept}{\textsf{accept}\xspace}

\newtheorem{lemma}{Lemma}[section]
\newtheorem{claim}[lemma]{Claim}

\newtheorem{fact}[lemma]{Fact}
\newtheorem{definition}[lemma]{Definition}

\newtheorem{observation}[lemma]{Observation}

\newtheorem{remark}[lemma]{Remark}

\usepackage{hyperref}
\hypersetup{
  colorlinks=true,
  linkcolor=red,
  citecolor=green,
  urlcolor=magenta
}



\usepackage[capitalise]{cleveref}
\crefname{problem}{Problem}{Problems}
\crefname{conjecture}{Conjecture}{Conjectures}
\crefname{fact}{Fact}{Facts}
\crefname{observation}{Observation}{Observations}
\crefname{claim}{Claim}{Claims}
\crefname{remark}{Remark}{Remarks}

\DeclareEmphSequence{\color{purple}\itshape}

\AtBeginEnvironment{algorithm}{%
  \DeclareEmphSequence{\itshape}
}
\AtEndEnvironment{algorithm}{%
  \DeclareEmphSequence{\color{purple}\itshape}
}

\AtBeginEnvironment{thebibliography}{%
  \DeclareEmphSequence{\itshape}
}
\AtEndEnvironment{thebibliography}{%
  \DeclareEmphSequence{\color{purple}\itshape}
}

\usepackage[colorinlistoftodos,prependcaption,textsize=tiny]{todonotes}

\newcommand{\mtt}[1]{\todo[backgroundcolor=orange!45]{MT: #1}}
\newcommand{\yx}[1]{\todo[backgroundcolor=gray!45]{yx: #1}}

\begin{document}


\maketitle
\begin{abstract}
    \input{abstract}

\end{abstract}
\thispagestyle{empty}
\newpage
\pagenumbering{arabic}

\input{intro}

\input{prel}


\input{k-ec-upper}

\input{k-vc-upper}
\input{lower_v2}

\input{2-conn_v2}

\input{conclusion_v2}

\section*{Acknowledgment}
 
ChatGPT was used during the research stage to assist with identifying relevant literature. It was also used during manuscript preparation to polish text and refine figures. The AI tool was not used to generate the exposition of the paper.

\printbibliography


\end{document}

%% file: abstract.tex
Local certification is a framework for verifying global graph properties using only local information. In this model, a prover assigns short labels, called certificates, to the vertices of a graph. Each vertex then exchanges certificates with its neighbors and performs a purely local check to determine whether the entire graph satisfies the desired property. This line of research has led to efficient certification schemes for a broad spectrum of graph classes, including minor-closed families, topological graph classes, and graphs defined by forbidden subgraphs.

In this paper, we study the local certification of graph connectivity, a classical problem that captures network robustness and is widely studied in distributed computing. Prior work by Bousquet, Feuilloley, and Pierron (JPDC 2024) showed that $2$-vertex-connectivity, $2$-edge-connectivity, and $3$-vertex-connectivity admit $O(\log n)$-bit certificates, leveraging structural characterizations such as ear decompositions. We go substantially beyond these cases and investigate general $k$-vertex-connectivity and $k$-edge-connectivity. We develop new approaches that exploit further connections between connectivity and combinatorial structures, including branchings, Eulerian subgraphs, and independent spanning trees. 
Our results are as follows. 
\begin{itemize}
    \item We obtain a tight $\Theta_k(\log n)$ bound for certifying $k$-edge-connectivity for every $k\ge 3$: we give an $O_k(\log n)$-bit certification scheme and prove a matching $\Omega_k(\log n)$ lower bound. Our lower bound applies to both $k$-edge-connectivity and $k$-vertex-connectivity. 
    \item  For $k$-vertex-connectivity, we obtain $\tilde{O}_k(\sqrt{n})$-bit certificates for every $k$ under a conjecture of Itai and Zehavi.
    \item We further show that, for $k=2$, the logarithmic barrier can be broken on sparse graph classes: $2$-edge-connectivity admits constant-size certificates in bounded-expansion graphs, and $2$-vertex-connectivity admits constant-size certificates in bounded-degree graphs. In contrast, for $2$-vertex-connectivity in general graphs, we prove an $\Omega(\log(\log^\ast n))$-bit lower bound.
\end{itemize}



%% file: intro.tex
\section{Introduction}


We study local certification of vertex and edge connectivity in distributed networks, modeled as $n$-vertex simple undirected graphs $G$. Fix a graph property $\mathcal{P}$ to be certified. In a \emph{local certification scheme}~\cite{KKP05}, each vertex of $G$ has a distinct $O(\log n)$-bit identifier, and a centralized prover assigns a short label to every vertex. Each vertex then executes a deterministic local verification algorithm, inspecting only its own identifier and label, together with the labels of its neighbors, and outputs either \accept\ or \reject.

The scheme must satisfy the following completeness and soundness guarantees:
\begin{description}
\item[Completeness:] If $G$ satisfies $\mathcal{P}$, then there exists a labeling under which every vertex accepts.
\item[Soundness:] If $G$ does not satisfy $\mathcal{P}$, then under every possible labeling, at least one vertex rejects.
\end{description}

The efficiency of a scheme is measured by its \emph{per-vertex label size}, defined as the maximum number of bits assigned to any vertex by the prover. As is the norm in this model, we assume that the underlying network $G$ is connected. For each integer $k\ge 2$, we investigate the per-vertex label size required to certify $k$-edge-connectivity and $k$-vertex-connectivity.

A graph is \emph{$k$-edge-connected} if every pair of vertices is connected by at least $k$ edge-disjoint paths. Similarly, a graph is \emph{$k$-vertex-connected} if every pair of vertices is connected by at least $k$ internally vertex-disjoint paths.

Certifying edge or vertex connectivity is a natural preprocessing step whenever connectivity is used as an algorithmic resource. A local certification scheme can turn this global structural property into a succinct, locally verifiable guarantee, allowing subsequent procedures to safely rely on a prescribed level of connectivity.

One example is the work of~\citet{BFP24}. They showed that $k$-vertex-connectivity, for $k\le 3$, can be certified using an  $O(\log n)$-bit label per vertex, and further developed local certification schemes for decompositions into $2$- or $3$-vertex-connected components. The key point is that higher connectivity can dramatically simplify the structure of $H$-minor-free graphs: for some small excluded minors $H$, highly connected $H$-minor-free graphs admit much simpler structural descriptions, which can then be exploited in the design of local certification schemes. Building on this idea,~\citet{BFP24} obtain $O(\log n)$-bit certification schemes for several $H$-minor-free graph classes. This fits into the broader program of locally certifying minor-closed graph classes~\cite{CKM25,esperet2022local,FFM+20,FFM+23}, which include fundamental families such as planar graphs and bounded-treewidth graphs. 

The next challenge is to go beyond $3$-vertex-connectivity. Indeed,~\citet{BFP24} explicitly asked for efficient local certification schemes for higher vertex connectivity~\cite[Open Question~58]{BFP24}:

\begin{center}
\fbox{\parbox{0.72\linewidth}{
\centering
\emph{Can $k$-vertex-connectivity be certified with $O(\log n)$ bits for any $k\ge 4$?}
}}
\end{center}

Such schemes could open the door to extending this approach to more complicated minor-closed classes and larger excluded minors. At the same time, it is equally important to understand whether logarithmic label size is necessary. In this direction,~\citet{BFP24} posed the following question~\cite[Open Question~57]{BFP24}:

\begin{center}
\fbox{\parbox{0.85\linewidth}{
\centering
\emph{Does the certification for $k$-vertex-connectivity require $\Omega(\log n)$ bits for any $k \ge 2$?}
}}
\end{center}

Despite graph connectivity being one of the most fundamental graph properties, the complexity of its local certification remains largely open, for both upper and lower bounds.

Connectivity is also a key resource in the design of distributed algorithms themselves. For example,~\citet{ChandraCD0L24} use edge-disjoint spanning trees guaranteed by high edge connectivity to disseminate information efficiently. Their algorithm solves $q$-message broadcast in a network $G$ of edge connectivity $\lambda$ in $O(((n+q)/\lambda)\log n)$ rounds. Since the vertices do not know $\lambda$ in advance, the algorithm uses a guessing procedure to obtain a suitable estimate. A local certificate for $k$-edge-connectivity suggests another approach: the network can first verify that its edge connectivity is at least $k$, and then run an algorithm under this certified guarantee.

The usefulness of high connectivity extends well beyond information dissemination. High \emph{edge connectivity} has been exploited in distributed algorithms for fault-tolerant and secure computation~\cite{censor2018fast, DBLP:conf/icalp/ChuzhoyPT20, DBLP:conf/wdag/HitronPY22}, and a separate line of work studies algorithms whose guarantees depend explicitly on the \emph{vertex connectivity} of the underlying graph~\cite{DBLP:conf/podc/Censor-HillelGK14,CHGG+17,CHGK13}.

 

\subsection{Our results}
We begin by showing that, for every $k \ge 3$, certifying $k$-vertex-connectivity or $k$-edge-connectivity requires per-vertex labels of size $\Omega((\log n)/k)$ (\cref{thm: k-conn-lower}). This answers \cite[Open Question~57]{BFP24} in the affirmative for every constant $k \ge 3$. Moreover, for constant $k$, the lower bound already holds on \emph{bounded-degree graphs}, that is, graphs whose maximum degree $\Delta$ is bounded by a constant.

Interestingly, the situation changes for $k=2$. Our lower-bound argument does not extend to this case, and this turns out to reflect a genuine change in complexity: as we show in \cref{thm: 2-edge-conn} and \cref{thm: 2-vertex-conn}, on bounded-degree graphs, both $2$-edge-connectivity and $2$-vertex-connectivity can be certified using only a constant number of bits per vertex.

\begin{restatable}{theorem}{thmKConnLower}\label{thm: k-conn-lower}
    For every integer $k \ge 3$, any local certification scheme for $k$-edge-connectivity or for $k$-vertex-connectivity requires per-vertex labels of size $\Omega((\log n)/k)$, even when each vertex has access to the identifiers of its neighbors. The lower bound holds in bounded-degree graphs when $k$ is a constant.
\end{restatable}

Our second result is a matching upper bound for $k$-edge-connectivity (\cref{thm: k-ec}). Together with \cref{thm: k-conn-lower}, this gives tight bounds for certifying $k$-edge-connectivity for every constant $k \ge 3$.

The certification scheme in \cref{thm: k-ec} uses a unified construction for all $k \ge 2$. The key idea is to move to a directed setting: we replace each edge $\{u,v\}$ of $G$ with the two antiparallel arcs $(u,v)$ and $(v,u)$, obtaining a directed graph $\vec{G}$, and apply \emph{Edmonds' branching theorem}~\cite{Edmonds73}. This yields a single construction that works for arbitrary $k$.

This contrasts with prior work~\cite{BFP24}, which uses \emph{ear decompositions} and their variants to certify $2$-edge-connectivity, $2$-vertex-connectivity, and $3$-vertex-connectivity. While effective for these small values of $k$, that approach does not appear to extend naturally to arbitrary $k$. The directed viewpoint is what makes a uniform construction possible.

\begin{restatable}{theorem}{thmkEC}\label{thm: k-ec}
    For every integer $k \ge 2$, there is a local certification scheme for $k$-edge-connectivity with per-vertex label size $O(k\log n)$.
\end{restatable}

Our third result is an $O(k^3\sqrt{n}\log^{2.5} n)$ upper bound for $k$-vertex-connectivity (\cref{thm: k-vc-cert}). The result is unconditional for $k=4$, and for every $k \ge 5$ it holds assuming the \emph{Itai--Zehavi conjecture}~\cite{ZI89}. This conjecture states that an undirected graph $G$ is $k$-vertex-connected if and only if, for every choice of vertex $r\in V(G)$, there exist $k$ \emph{independent spanning trees} rooted at $r$. Here, independence means that, for every $x\in V(G)\setminus\{r\}$, the $r$-to-$x$ paths in the $k$ trees are internally vertex-disjoint.

The best previous result~\cite{BFP24} gives $O(\log n)$-bit local certification schemes for $2$- and $3$-vertex-connectivity, using open ear decompositions and Mondshein sequences, while the case $k\ge 4$ was left open as \cite[Open Question~58]{BFP24}. We make progress on this open problem by giving the first local certification schemes for $k\ge 4$: unconditionally for $k=4$, and for all $k\ge 5$ under the Itai--Zehavi conjecture.


\begin{restatable}{theorem}{thmKVC}\label{thm: k-vc-cert}
    For every integer $k\ge 4$ for which the Itai--Zehavi conjecture holds, $k$-vertex-connectivity admits a local certification scheme with per-vertex label size $O(k^3\sqrt{n}\log^{2.5} n)$. In particular, since the conjecture is known to hold for $k=4$, $4$-vertex-connectivity admits a local certification scheme with per-vertex label size $O(\sqrt{n}\log^{2.5} n)$.
\end{restatable}


Our lower bound in \Cref{thm: k-conn-lower} applies only for $k\ge 3$, leaving open a natural question: what is the complexity of certifying $2$-edge-connectivity and $2$-vertex-connectivity? In particular, can one beat the logarithmic barrier, or must the label size still grow with $n$?

We show that, on bounded-degree graphs, both properties admit local certification schemes using only a constant number of bits per vertex. For $2$-edge-connectivity, we go further and extend the constant-size certification to \emph{bounded-expansion graphs}, a broad family that includes bounded-degree graphs, planar graphs, and, more generally, every proper minor-closed graph class.

Together with \Cref{thm: k-conn-lower}, these results reveal a sharp complexity gap between $k=2$ and $k\ge 3$ on bounded-degree graphs.


\begin{restatable}{theorem}{thmTwoConn}\label{thm: 2-conn}
    For $2$-connectivity in restricted graph classes, we obtain the following local certification schemes with constant-size per-vertex labels. All three schemes work even in the {anonymous} setting where vertices have no identifiers.
    \begin{itemize}
        \item There exists a local certification scheme for $2$-edge-connectivity with per-vertex label size $O(\Delta\log\Delta)$.
        \item There exists a local certification scheme for $2$-vertex-connectivity with per-vertex label size $O(2^{\Delta}\cdot\Delta\log\Delta)$.
        \item Let $\Fcal$ be any bounded-expansion graph class. Given the promise that the input graph belongs to $\Fcal$, there exists a local certification scheme for $2$-edge-connectivity with $O(1)$ per-vertex label size, where the constant depends on the graph class $\Fcal$.
    \end{itemize}
\end{restatable}

The restriction to sparse graph classes is essential for these constant-size upper bounds. We prove a super-constant lower bound for certifying $2$-vertex-connectivity on general graphs in the anonymous setting, via a slight modification of the very recent parity certification lower bound of~\citet{BFVZ26}. The construction inherently exploits the absence of unique identifiers and therefore does not extend to the setting with unique identifiers.

\begin{restatable}{theorem}{thmTwoVCLB}\label{thm: 2-vc-lb}
    Any local certification scheme for $2$-vertex-connectivity requires per-vertex labels of size $\Omega(\log(\log^*n))$ when the vertices are anonymous.
\end{restatable}


We summarize the known and new upper and lower bounds on the optimal per-vertex label size for certifying $k$-edge-connectivity and $k$-vertex-connectivity in~\Cref{tab: results1} and~\Cref{tab: results2}.

\input{result-table}
\subsection{Technique overview}

We present a technical overview of our proofs.

\paragraph{Lower-bound technique for $k\ge 3$.}
Our lower bound follows the cut-and-plug framework of Göös and Suomela~\cite{GS16}. The graph-theoretic task is to find, for each property $\mathcal{P}$, a witness graph $G\in\mathcal{P}$ with a small local interface whose rewiring destroys the property. In our setting, this interface is an edge $\{u,v\}$ whose endpoints have only $O(k)$ neighbors. Consider two copies $G_1$ and $G_2$ of $G$, with corresponding edges $\{u_1,v_1\}$ and $\{u_2,v_2\}$. Both $G_1$ and $G_2$ satisfy $\mathcal{P}$, so there are labelings under which all vertices accept. The cut-and-plug operation starts from the vertex-disjoint union $G_1 \uplus G_2$, deletes the two edges $\{u_1,v_1\}$ and $\{u_2,v_2\}$, and adds the cross edges $\{u_1,v_2\}$ and $\{u_2,v_1\}$. For our witness graphs, this edge-swapped union no longer satisfies $\mathcal{P}$: in fact, its edge connectivity and vertex connectivity are both at most $2$. However, if the labels are too short, the two copies can be chosen so that, after this edge swap, every vertex still sees exactly the same local view as it saw in some accepting copy. Hence all vertices still accept the edge-swapped union, contradicting soundness.

The rewiring itself uses only two copies, but the proof must first force two copies to be locally indistinguishable to the verifier. To do this, we consider polynomially many copies of the same witness graph, with different assignments of identifiers to the vertices. If the labels have size $o(\log n)$, then there are too few possible labeled local views around the small interface to distinguish all these copies. Hence two copies can be chosen so that, after the rewiring, every vertex still sees a local view that it would accept in some valid copy. The verifier is therefore forced to accept the rewired graph, even though it is not $k$-edge-connected or $k$-vertex-connected for any $k\ge 3$, contradicting soundness.

This also explains why the strategy is tight with respect to $k=2$: the obstruction produced by the rewiring has connectivity at most $2$, which is not enough to refute $2$-connectivity. Indeed, on bounded-degree graphs, both $2$-edge-connectivity and $2$-vertex-connectivity admit $O_\Delta(1)$-bit certificates.

\paragraph{Upper-bound technique for edge connectivity.}
For the upper bound on $k$-edge-connectivity, we move from the undirected graph $G$ to the directed graph $\vec{G}$ obtained by replacing each edge $\{u,v\}$ with the two antiparallel arcs $(u,v)$ and $(v,u)$. The key observation is that $G$ is $k$-edge-connected if and only if $\vec{G}$ is $k$-arc-strongly connected. We can then apply Edmonds' branching theorem~\cite{Edmonds73}, which witnesses high arc-connectivity through collections of arc-disjoint spanning branchings rooted at a common vertex. The prover certifies these branching structures, and each vertex locally checks their consistency.

A key advantage of this approach is that the same construction works for every $k$: increasing $k$ simply requires certifying more branching structures. Thus, the directed formulation provides a uniform way to certify $k$-edge-connectivity for arbitrary $k$.


\paragraph{Upper-bound technique for vertex connectivity.}
For the upper bound for $k$-vertex-connectivity, we use the Itai--Zehavi characterization in terms of independent spanning trees. Assuming this characterization for the given value of $k$, it suffices to certify suitable collections of rooted spanning trees such that, for every vertex $x$, the root-to-$x$ paths in the trees are internally vertex-disjoint. Thus the main building block is to certify that two prescribed spanning trees $T_1$ and $T_2$, rooted at the same vertex, are independent.

A naive certificate would give each vertex its full list of ancestors in both trees, but this can require linear space. We reduce the cost using heavy-light decompositions. In each tree, every root-to-vertex path intersects only $O(\log n)$ heavy-light segments, so the prover can encode such a path by a short list of segment-rank pairs. These encodings allow vertices to locally test ancestor relations. The remaining challenge is to ensure that, if independence fails, then some vertex can detect a common internal ancestor on its two root paths. To achieve this, the certificates propagate compact path addresses through the two trees, using a long/short segment threshold to balance the amount of information stored at each vertex: there are few long segments globally, while only few short-segment vertices appear along any root-to-vertex path. This yields an $\widetilde{O}(\sqrt{n})$-bit certificate for independence of two rooted spanning trees.

We also show that this building block is essentially optimal. Namely, there are $n$-vertex graphs with two prescribed spanning trees $T_1$ and $T_2$ rooted at the same vertex $r$ such that any local certification scheme for verifying that $T_1$ and $T_2$ are independent requires per-vertex label size $\Omega(\sqrt{n})$.

\begin{restatable}{proposition}{thmIndepST}\label{thm: indep-st-lower}
There are $n$-vertex graphs $G$ with two spanning trees $T_1$ and $T_2$ rooted at the same vertex $r$, given as part of the input, such that any local certification scheme for verifying that $T_1$ and $T_2$ are independent requires per-vertex label size $\Omega(\sqrt{n})$.
\end{restatable}

This independent-tree formulation makes the extension to larger $k$ direct. Once we can certify independence of two rooted spanning trees, we certify pairwise independence among the $k$ trees required for each root. The only additional cost is the number of tree pairs and roots, leading to the stated $O(k^3\sqrt{n}\log^{2.5} n)$ per-vertex label size.

\paragraph{Constant upper bound for $ k = 2$.}
Our main building block for certifying $2$-edge-connectivity is a decomposition theorem stating that every $2$-edge-connected graph can be covered by a constant number of (possibly disconnected) Eulerian subgraphs. This decomposition immediately suggests a certification strategy. Suppose that the prover is allowed to place labels not only on vertices but also on edges, where the label of an edge can be read by both of its endpoints. The prover can then assign a distinct label to each Eulerian subgraph and let every edge inherit the label of the subgraphs containing it. Each vertex then checks that, for every label, the number of incident edges carrying that label is even. On bounded-degree graphs, it is known that labels on edges can be transformed into labels on vertices with only an $O_\Delta(1)$ overhead~\cite{BFVZ26}. To extend this result to bounded-expansion graphs, the proof first extracts a suitable core structure for $2$-edge-connectivity and then transforms the edge-labeling on this core into a vertex-labeling.

For $2$-vertex-connectivity, we obtain constant-size per-vertex labels on bounded-degree graphs via a reduction to $2$-edge-connectivity. We replace each vertex by an edge and partition its neighbors into two nonempty sets, each attached to a different endpoint of the edge. The key property is that the original graph is $2$-vertex-connected if and only if every graph obtained in this manner is $2$-edge-connected. This equivalence is specific to $2$-connectivity: the converse direction fails for higher connectivity, and our reduction therefore does not generalize to certifying $k$-vertex-connectivity on bounded-degree graphs for $k\ge 3$.

\subsection{Additional related works}

Local certification, and more generally distributed interactive proofs, of various graph classes has attracted
significant attention in recent years: bounded-pathwidth graphs~\cite{baterisna2025optimal}, bounded-treewidth graphs~\cite{CKM25,fraigniaud2024meta}, planar graphs and their variants~\cite{FFM+20,gil_et_al:LIPIcs.DISC.2025.34,gil2026distributed,NPY20}, bounded-genus graphs~\cite{FFM+23}, cographs~\cite{montealegre2021compact}, distance-hereditary graphs~\cite{montealegre2021compact}, and geometric intersection graphs~\cite{jauregui2022distributed}.

\citet{GS16} introduced common complexity classes in local certification and prove tight bounds for several graph properties. For example, $s$--$t$ reachability can be certified in $O(1)$ bits, and to certify that $s$--$t$ connectivity equal to $k$, it suffices to use $O(\log k)$ bits; the cut-and-plug technique is developed to prove $\Omega(\log n)$ lower bound for many properties such as parity of the graph size and spanning trees, together with the existing upper bound, the complexity of these problems are $\Theta(\log n)$. There are also graph properties with certification complexity $\poly(n)$, such as symmetricity and non-$3$-colorability, where the lower bounds are reduced from communication complexity. 



There is a long line of work that studies the connectivity problem in the \textsf{CONGEST} model. Computing the min-cut in a network requires $\Tilde{\Omega}(\sqrt{n}+\diam)$ rounds~\cite{SHK+11, GK13}, where the lower bound also holds for any constant-approximation. In~\cite{GK13}, they proposed a $(2+\epsilon)$-approximation algorithm that almost matches the lower bound; the approximation ratio was later improved to $(1+\epsilon)$ by~\cite{NS14}. For exact min-cut,~\cite{GNT20} improved the complexity of~\cite{DHNS19} and provided the first sublinear algorithm; the worst-case optimal round complexity $\Tilde{O}(\sqrt{n}+\diam)$ was achieved in~\cite{DEM+21}, matching the lower bound. Later on, a universal optimal algorithm for exact min-cut was proposed by~\cite{GZ22}. The min-cut problem is also studied in other related models, such as congested clique~\cite{GN18} and massively parallel computation~\cite{GN20}.


\subsection{Paper organization}
The rest of the paper is organized as follows. In~\cref{sec:Pre}, we discuss variants of the local certification models and review several standard certification techniques. In \cref{sec:kec-upper}, we present our upper bound for $k$-edge-connectivity via Edmonds' branching theorem. In \cref{sec:kvc-upper}, we present our upper bound for $k$-vertex-connectivity via independent spanning trees. We then prove our lower bounds for certifying $k$-edge-connectivity and $k$-vertex-connectivity in \cref{sec: lower}. 
In~\Cref{sec: 2-conn-bdd-deg}, we study 2-edge-connectivity and 2-vertex-connectivity. Finally, in~\Cref{sec: conclusion}, we conclude with a discussion of future directions.

%% file: result-table.tex
\begin{table}[!h]
\centering
\begin{tabular}{m{17em}m{12em}m{10em}}

    Task & Upper and Lower Bounds & Reference\\ 
    \midrule
    \midrule
    
    \multirow{2}{16em}{$k$-edge-connectivity, $k\geq 3$}
     & $\Omega_k(\log n)$
     & \Cref{thm: k-conn-lower}\\
     & $O_k(\log n)$ 
     & \Cref{thm: k-ec}  \\
    \midrule

    $k$-vertex-connectivity, $k \ge 3$ & $\Omega_k(\log n)$ 
    & \Cref{thm: k-conn-lower} \\
    \addlinespace[8pt]
    $3$-vertex-connectivity & $O(\log n)$  & \cite{BFP24} \\
    \addlinespace[8pt]
    $4$-vertex-connectivity & $\tilde{O}_k(\sqrt{n})$ 
    & \Cref{thm: k-vc-cert} \\
    \addlinespace[8pt]
    \begin{tabular}[c]{@{}l@{}}
    $k$-vertex-connectivity, $k\geq 5$\\
    {\footnotesize (under the Itai--Zehavi conjecture)}
    \end{tabular}
    & $\tilde{O}_k(\sqrt{n})$  
    & \Cref{thm: k-vc-cert}\\
    \addlinespace[8pt]
    Two independent spanning trees & $\Omega(\sqrt{n})$ & \Cref{thm: indep-st-lower}\\
    \bottomrule
\end{tabular}
\caption{Per-vertex label size for $k$-edge-connectivity and $k$-vertex-connectivity.}
\label{tab: results1}

\end{table}

\begin{table}[!h]
    \centering
    \begin{tabular}{m{9em}m{9em}m{12em}m{9em}}

        Task & Graph Class & Upper and Lower Bounds & Reference\\ 
        \midrule
        \midrule

        \multirow{3}{16em}{$2$-edge-connectivity}
         & General & $O(\log n)$ & \cite{BFP24} \\ 
         & Max-degree $\Delta$ & $O_{\Delta}(1)$ & \Cref{thm: 2-edge-conn}\\ 
         & Bounded-expansion & $O(1)$ & \Cref{thm: 2-ec-bdd-expansion} \\ 
        \midrule
        
        \multirow{3}{16em}{$2$-vertex-connectivity}
         & General & $O(\log n)$ & \cite{BFP24} \\
         & Max-degree $\Delta$ & $O_{\Delta}(1)$ 
         & \Cref{thm: 2-vertex-conn} \\
         & General & $\Omega\left(\log(\log^*n)\right)$ & \Cref{thm: 2-vc-lb} \\
        \bottomrule
    \end{tabular}
    \caption{Per-vertex label size for $2$-edge- and $2$-vertex-connectivity in the anonymous model. 
    }
    \label{tab: results2}
\end{table}

%% file: prel.tex
\section{Preliminaries}\label{sec:Pre}

In this section, we discuss variants of the local certification models and review several standard certification techniques.

\paragraph{Notation.}
Throughout the paper, we use $\tilde{O}(\cdot)$, $\tilde{\Omega}(\cdot)$, and $\tilde{\Theta}(\cdot)$ to suppress factors of the form $\log^{\pm O(1)} n$, where $n$ denotes the number of vertices in the original graph $G$. We also use a subscript $k$ to suppress factors depending only on $k$, writing $O_k(\cdot)$, $\Omega_k(\cdot)$, and $\Theta_k(\cdot)$. For example, the bound $O(k^3 \sqrt{n} \log^{2.5} n)$ can be expressed as $\tilde{O}_k(\sqrt{n})$.

For convenience, throughout the paper, the size complexity (or bit complexity) of a local certification scheme refers to the maximum per-vertex (or per-element, defined later) label size.

\subsection{Local certification schemes}\label{sec:Pre-LocalCert}
Our local certification scheme model is precisely the ID-based \emph{proof-labeling scheme} of~\citet{KKP05}. In this model, the local verifier at each vertex $v$ has access to $\ID(v)$, any prescribed input at $v$, and the labels assigned by the prover to $v$ and its neighbors. Certification of $k$-vertex-connectivity or $k$-edge-connectivity requires no input. In contrast, for tasks such as certifying whether a given edge set forms a spanning tree, the edge set is provided as input by letting each vertex know which of its incident edges belong to the set.

We also consider two natural variants of the local certification scheme. In the first variant, the local verifier at each vertex additionally has access to the identifiers of its neighbors. This is exactly the model of \emph{locally checkable proofs} with local horizon~$1$~\cite{GS16}. In the second variant, vertices have no identifiers at all; this is known as the \emph{anonymous model}~\cite{KKP05,BFZ25}.

These variants are naturally ordered by their computational power. Any upper bound established in the anonymous model immediately carries over to the ID-based setting. Likewise, an upper bound without access to neighbors' identifiers also applies when such access is available. For lower bounds, the implications are reversed.

In the ID-based setting, however, access to neighbors' identifiers does not affect the asymptotic complexity once the maximum label size is $\Omega(\log n)$. Indeed, the prover can simply include each vertex's identifier in its label, and the vertex can locally verify its consistency.

In contrast, the anonymous and ID-based settings are fundamentally different. Several tasks, including certifying leader election and spanning trees, are impossible in the anonymous model. On the other hand, in the ID-based setting, every computable function admits a certification scheme by assigning every vertex a description of the entire network. See~\cite{KKP05} for further discussion.

\subsubsection{Edge certification and degeneracy}\label{sec:Pre-EdgeCert}
In the \emph{edge certification scheme}~\cite{BFP24, FFM+23}, the setting is the same as in the standard local certification scheme, except that the prover can assign labels to both vertices and edges. The label of an edge can be read by both of its endpoints. The complexity is measured by the maximum \emph{per-element} label size, which is the maximum of the per-vertex label size and the per-edge label size. To distinguish the two models, we sometimes call the standard local certification schemes, where certificates are assigned only to vertices, the \emph{vertex-certification schemes}.

Edge-certification schemes are a commonly used building block for designing vertex-certification schemes. In the following, we state two well-known transformations from edge-certification schemes to vertex-certification schemes on restricted graph classes. See~\cite{BFP24, FFM+23} for further applications.

For bounded-degree graphs with maximum degree at most $\Delta$, we can transform any edge certification to a vertex certification with only a constant blow-up. The idea is to do a $2$-hop coloring on the vertices, and for each vertex, assign a copy of the certificate on each incident edge, along with the color on the other endpoint of that edge. Then each vertex knows which edge each certificate belongs to. See also~\cite[Proposition 4]{BFVZ26}.
\begin{lemma}\label{lem: edge-cert-bdd-deg}
    Consider an edge certification of size $f(n)$ for a graph family $\Fcal$. If $\Fcal$ has bounded degree $\Delta$, then there exists a (vertex) certification of size $O(\Delta (f(n)+\log\Delta))$.
\end{lemma}

A similar result has been established for sparse graphs with bounded degeneracy. The \emph{degeneracy} of a graph $G$ is defined to be the minimum number $k$ where $G$ admits an ordering of the vertices such that every vertex has at most $k$ neighbors that are earlier in the ordering. In particular, if a graph is of bounded degeneracy, then the number of edges is linear in the number of vertices. Such an ordering induces an orientation of the edges in which every vertex has bounded out-degree. After certifying this orientation, each vertex is responsible for carrying the certificates of its outgoing edges. In this way, the edge labels can be redistributed uniformly among the vertices. The following lemma formalizes this transformation from edge-certification schemes to vertex-certification schemes on bounded-degeneracy graphs.

\begin{lemma}[{\cite{FFM+23}}]\label{lem:EdgeCert-degeneracy}
    Consider an edge certification scheme of size $f(n)$ for a graph family $\Fcal$. If every graph in $\Fcal$ has degeneracy at most $d$, then there is a (vertex) certification of size $O(d ( f(n)+\log n))$.
\end{lemma}

\subsection{Certifying BFS spanning tree}
Spanning tree certification is a widely used building block in local certification schemes. In such a scheme, the labels encode the parent--child relationships of a spanning tree and guarantee that either these labels describe a valid spanning tree or at least one vertex outputs \reject.
\begin{lemma}[{\cite{KKP05}}]\label{lem: BFS-cert}
    Let $G$ be a connected graph. Certifying a BFS spanning tree can be done in $O(\log n)$ bits.
\end{lemma}
\begin{proof}
For completeness, we briefly describe the certification scheme and its verification procedure.
    \begin{itemize}
        \item Pick any vertex $r\in V$ to be the root. For each vertex $v\in V$, the certificate $c(v)$ contains three fields $(r, p(v), \dist(v,r))$, where $r$ is the root's $\ID$, $p(v)$ is the parent's $\ID$, and $\dist(v,r)$ is the distance from $v$ to root $r$.
        \item The verification process contains the following parts:
        \begin{itemize}
            \item (All distances are consistent.) Each vertex $v$ checks that the distance from each of its neighbors to the root differs from itself by at most $1$. That is, for each $u\in N(v)$, $\dist(u,r)\in\{\dist(v,r)-1,\dist(v,r),\dist(v,r)+1\}$. 
            \item (There exists a root.) Each vertex $v$ with $\dist(v,r)>0$ checks that its parent is closer to the root. This can be done using $p(v)$ in the certificate of $v$, and $\dist(\cdot,v)$ in the certificates of its neighbors. 
            \item (The root is unique.) Each vertex checks that all its neighbors receive the same root $\ID$. For vertex $v$ such that $\dist(v,r)=0$, $v$ itself is the root. Therefore, we have $v$ check that $r=\ID(v)$.\qedhere
        \end{itemize}
    \end{itemize}
\end{proof}
Using this scheme, one can certify a spanning forest on a (not necessarily connected) graph $G$ by certifying a spanning tree on each connected component of $G$.

%% file: k-ec-upper.tex
\section{Certifying $k$-edge-connectivity in $O_k(\log n)$ bits}\label{sec:kec-upper}


In this section, we prove the following theorem. 

\thmkEC*

For every integer $k \ge 1$, we use the equivalence stated in \cref{lem: ec-directed} to certify that a simple undirected graph $G$ is $k$-edge-connected. \cref{lem: ec-directed} may be considered as a variant of the strong orientation theorem~\cite{Robbins39,NashWilliams60}.

\begin{lemma}\label{lem: ec-directed}
Let $k \ge 1$ be an integer. A simple undirected graph $G=(V,E)$ is $k$-edge-connected if and only if its bidirected version $\vec{G}=(V,A)$ is $k$-arc-strongly connected, where $\vec{G}$ is obtained from $G$ by replacing each undirected edge $\{u,v\}\in E$ with the two antiparallel arcs $(u,v)$ and $(v,u)$.
\end{lemma}
\begin{proof}
($\Rightarrow$) Let $x,y\in V$ be distinct vertices. Since $G$ is $k$-edge-connected, there are at least $k$ edge-disjoint undirected paths in $G$ from $x$ to $y$. These paths witness the existence of at least $k$ arc-disjoint directed paths in $\vec{G}$ from $x$ to $y$ because each undirected edge in $G$ corresponds to two antiparallel arcs in $\vec{G}$.

($\Leftarrow$)
Let $x,y\in V$ be distinct vertices. Consider any cut $(X,Y)$ separating $x$
from $y$, with $x\in X$ and $y\in Y$. Since $\vec{G}$ is $k$-arc-strongly
connected, the directed version of Menger's theorem~\cite{Menger27} implies that there are at
least $k$ arcs from $X$ to $Y$ in $\vec{G}$. These arcs correspond to distinct
edges of $G$ crossing the cut $(X,Y)$, because the antiparallel arc of an arc
from $X$ to $Y$ is directed from $Y$ to $X$. Hence every cut $(X, Y)$ in $G$ has
size at least $k$. By the undirected version of Menger's theorem~\cite{Menger27}, there are at
least $k$ edge-disjoint paths in $G$ from $x$ to $y$.
\end{proof}

Another key certificate used by our algorithm comes from Edmonds' branching theorem
(\cref{lem:branching}). We next explain how to combine \cref{lem: ec-directed}
with \cref{lem:branching} to obtain a local certification scheme for
$k$-edge-connectivity for every integer $k \ge 2$.

Suppose that $G$ is $k$-edge-connected. Then, by \cref{lem: ec-directed},
its bidirected version $\vec{G}$ is $k$-arc-strongly connected. Fix a vertex
$r \in V(\vec{G})$. By Edmonds' branching theorem, $\vec{G}$ contains a
collection of $k$ pairwise arc-disjoint \emph{out-branchings} rooted at $r$, as well
as a collection of $k$ pairwise arc-disjoint \emph{in-branchings} rooted at $r$, where an out-branching (resp., an in-branching) rooted at $r$ is a directed 
spanning tree whose arcs are all oriented away from $r$ (resp.,
toward $r$).
Thus, the prover can construct these two collections of branchings in $\vec{G}$
and certify their existence to the verifiers.

It remains to see why these branchings certify $k$-edge-connectivity of $G$.
Assume that the verifiers are convinced that $\vec{G}$ contains the two
collections of branchings described above. We claim that every nontrivial
cut\footnote{A cut $(S,V(G)\setminus S)$ is nontrivial if both $S$ and
$V(G)\setminus S$ are nonempty.} $(S,V(G)\setminus S)$ has at least $k$ arcs
entering $S$ in $\vec{G}$. If $r \notin S$, then each out-branching contains at
least one arc entering $S$ from $V(G)\setminus S$. Since the $k$ out-branchings
are pairwise arc-disjoint, this gives at least $k$ arcs entering $S$. If
$r \in S$, the same argument applied to the $k$ in-branchings gives at least
$k$ arcs entering $S$. Hence every nontrivial cut has at least $k$ incoming
arcs.

Therefore, after deleting any set of at most $k-1$ arcs from $\vec{G}$, every
nontrivial cut still has at least one incoming arc. This implies that the
remaining digraph is strongly connected, and hence $\vec{G}$ is
$k$-arc-strongly connected. By \cref{lem: ec-directed}, $G$ is
$k$-edge-connected.

Conversely, if $G$ is not $k$-edge-connected, then \cref{lem: ec-directed}
implies that $\vec{G}$ is not $k$-arc-strongly connected. In this case, such
two collections of branchings cannot exist; otherwise, the argument above would
imply that $\vec{G}$ is $k$-arc-strongly connected. Hence no prover can certify
the existence of these two collections to the verifiers.

\begin{lemma}[Edmonds’ Branching Theorem~\cite{Edmonds73}]\label{lem:branching}
Let $k \ge 1$ be an integer.
Every digraph $\vec{G}=(V,A)$ contains $k$ arc-disjoint out-branchings rooted at a vertex $r$ if and only if, for every nonempty subset $S \subseteq V \setminus \{r\}$, at least $k$ arcs enter $S$ from $V \setminus S$.
\end{lemma}

Consequently, it remains for the prover to certify that, for some vertex $r \in V(G)$, there exist $k$ pairwise arc-disjoint out-branchings rooted at $r$. The prover must also provide an analogous certificate for $k$ pairwise arc-disjoint in-branchings rooted at $r$. In what follows, we focus on the out-branchings case.

\begin{itemize}
    \item Prover: Fix a vertex $r \in V(G)$ and compute $k$ arc-disjoint out-branchings of $\vec{G}$ rooted at $r$, denoted $T_i$ for $i \in [k]$. For each vertex $v \in V(G)$ and each $i \in [k]$, give the verifier for $v$ a certificate
    \[
        c_i(v) \coloneqq \left(\ID(r),\, p_i(v),\, \depth_i(v)\right),
    \]
    where $p_i(v)$ is the $\ID$ of $v$'s parent in $T_i$ (with $p_i(r)$ set to a null pointer) and $\depth_i(v)$ is the depth of $v$ in $T_i$. 
    \item The verification process contains the following parts:
        \begin{itemize}
            \item (The root $r$ is unique.) Each vertex checks that all of its neighbors report the same root identifier $\ID(r)$. Since $G$ is connected, if different roots were assigned, then some edge would connect two vertices with different root identifiers, and both of them would reject. The root $r$ also checks that $\depth_i(r)=0$ for all $i \in [k]$.
            
            \item ($T_i$ consists of $n-1$ edges.) Each vertex $v$ with $\depth_{i}(v)>0$ checks that its parent has a smaller depth. This can be done using $p_{i}(v)$ in the certificate of $v$, and $\depth_i(p_i(v))$ in the certificate of its parent $p_i(v)$. 

            \item ($T_i$ for $i \in [k]$ are pairwise arc-disjoint.) Each vertex $v \ne r$ checks whether the parent identifiers $p_i(v)$ over all $i \in [k]$ are pairwise distinct. If some arc $(x,y)$ were used in more than one branching, then $y$ would have the same parent $x$ in at least two of the $T_i$, and this violation would be detected by vertex $y$.
        \end{itemize}
    
    Given the certificates $c_i(v)$ for all $v \in V(G)$ and $i \in [k]$, the verifiers can locally check that, for each $i \in [k]$, all vertices agree on the same root and that every vertex other than the root has exactly one parent. Together with the depth consistency checks, it follows that the underlying undirected graph of $T_i$ is connected and has exactly $n-1$ edges, and hence is a tree. Since every arc of $T_i$ is oriented from parent to child, so $T_i$ is an out-branching rooted at $r$. Finally, the verifiers can also check that the $T_i$ are pairwise arc-disjoint. This yields an $O(k \log n)$-bit local certification scheme for $k$-edge-connectivity for every integer $k \ge 2$.
    
\end{itemize}

%% file: k-vc-upper.tex
\section{Certifying $k$-vertex-connectivity in $\Tilde{O}_k(\sqrt{n})$ bits}\label{sec:kvc-upper}
To explain our upper bound for $k$-vertex-connectivity, we first recall the following conjecture of Itai and Zehavi~\cite{ZI89}: \emph{a graph $G$ is $k$-vertex-connected if and only if for every vertex $r\in V(G)$, there exist $k$ independent spanning trees rooted at $r$.}

We also recall that independence means that for every vertex $x\in G\setminus r$, the paths from $r$ to $x$ in the $k$ trees are internally vertex-disjoint. See~\Cref{fig: IST-illustration} for an illustration when $k=2$.

The Itai--Zehavi conjecture is known to hold for every $k \in \{1, 2, 3, 4\}$~\cite{ItaiR88,CheriyanM88,ZI89, CLY06}. Using this characterization, we obtain a local certification scheme for $k$-vertex-connectivity with per-vertex label size $\tilde{O}_k(\sqrt{n})$ for $k\in\{2,3,4\}$. For $k\ge 5$, we obtain the same upper bound assuming the Itai--Zehavi conjecture.

\thmKVC*

We note that, by a simple observation, the Itai--Zehavi conjecture implies that a graph $G$ is $k$-vertex-connected if and only if there exist $k$ distinct vertices $r_1,\ldots,r_k \in V(G)$ (rather than all vertices) such that, for each $i \in [k]$, $G$ contains $k$ independent spanning trees rooted at $r_i$. Suppose that $G$ contains such roots but is not $k$-vertex-connected, then there exists a vertex cut $C$ with $|C|<k$ such that $G-C$ is disconnected. Since $C$ has size less than $k$, it cannot contain all of $r_1,\ldots,r_k$, so some root $r^*$ remains in $G-C$. Since $G-C$ is disconnected, there exists a vertex $x$ that belongs to a component of $G-C$ different from the one containing $r^*$. This contradicts the existence of $k$ independent spanning trees in $G$ rooted at $r^*$, which implies that there are $k$ internally vertex-disjoint paths between $r^*$ and $x$, and hence no vertex cut of size less than $k$ can separate them.

Given the above equivalence, certifying $4$-vertex-connectivity reduces to certifying that $\binom{4}{2}$ pairs of spanning trees rooted at a common vertex are independent. More generally, assuming the Itai--Zehavi conjecture, the same reduction applies to certifying $k$-vertex-connectivity for every $k \ge 5$. The task of certifying $k$-vertex connectivity is therefore reduced to the following building block: fix a root $r^*$ and certify that every pair among the $k$ spanning trees rooted at $r^*$ is independent. In the following lemma, we prove that $\Tilde{O}_k(\sqrt{n})$ bits suffice to certify two independent spanning trees rooted at the same vertex.


\begin{lemma}\label{lem: 2-indep-tree-certification}
    Given a graph $G$, a vertex $r\in G$, and two spanning trees $T_1$ and $T_2$ on $G$ rooted at $r$ (the spanning tree inputs are of the form of the certification scheme described in \Cref{lem: BFS-cert}), there exists an $O(\sqrt{n}\log^{2.5}n)$-bit local certification scheme for certifying that $T_1$ and $T_2$ are independent.
\end{lemma}
Given this lemma, \Cref{thm: k-vc-cert} follows directly.
\begin{proof}[Proof of \Cref{thm: k-vc-cert}]
    Pick arbitrary $k$ vertices $r_1,\ldots,r_k$ in $G$. For each $i\in[k]$, we certify $k$ independent spanning trees rooted at $r_i$ by applying \Cref{lem: 2-indep-tree-certification} for $\binom{k}{2}$ pairs of spanning trees rooted at $r_i$. Therefore, the complexity increases by a factor of $O(k^3)$.
\end{proof}
\Cref{sec: 2-indep-spt-cert} is devoted to proving \Cref{lem: 2-indep-tree-certification}. In~\Cref{sec:lower-2indst}, we show that this approach is essentially tight in the sense that given $2$ prescribed spanning trees, certifying that they are independent requires per-vertex label size $\Omega(\sqrt{n})$.

\begin{figure}[ht]
    \centering
    \begin{subfigure}{0.4\textwidth}
        \centering
        \includegraphics[width=\textwidth]{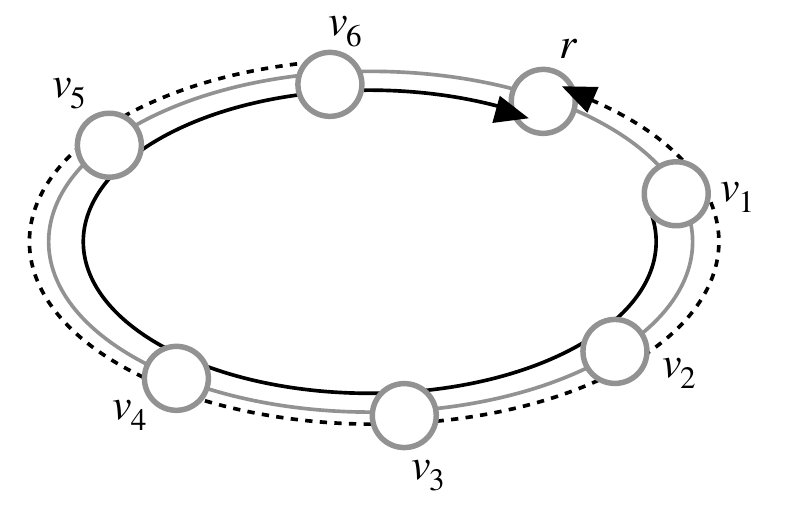}
        \caption{Two spanning trees that are vertex-independent.}
        \label{fig: IST}
    \end{subfigure}
    \hspace{1.5cm}
    \begin{subfigure}{0.4\textwidth}
        \centering
        \includegraphics[width=\textwidth]{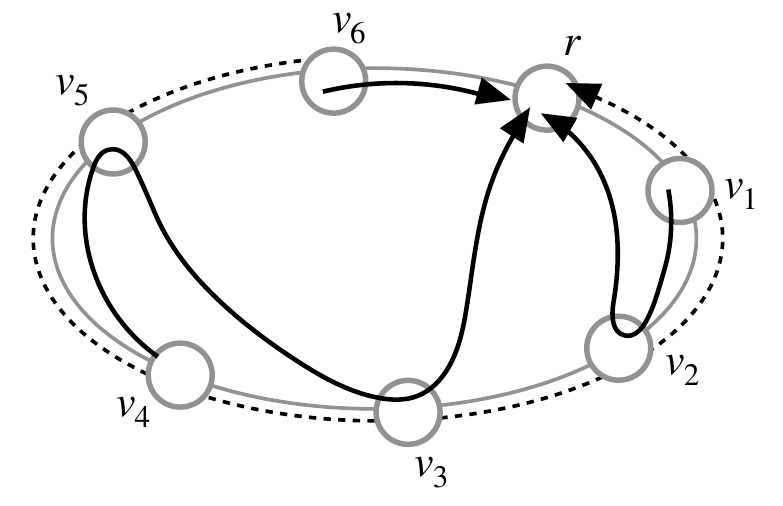}
        \caption{Two spanning trees that are not vertex-independent}
        \label{fig: non-IST}
    \end{subfigure}
    \caption{Illustration of independent spanning trees. Both figures demonstrate two spanning trees in a ring rooted at the same vertex $r$, one in black and the other in dotted lines, oriented from the leaves to the root. In~\Cref{fig: IST}, both spanning trees are paths; the dotted one has $v_6$ as the leaf, and the other has $v_1$ as the leaf. The two trees are vertex-independent. In~\Cref{fig: non-IST}, the two trees are not vertex-independent (but still edge-independent) since the paths from $v_4$ (and also $v_5$) pass through $v_3$ in both trees.}
    \label{fig: IST-illustration}
\end{figure}

\subsection{Proof of \Cref{lem: 2-indep-tree-certification}}\label{sec: 2-indep-spt-cert}
We first consider the following naive approach: Given two spanning trees $T_1,T_2$ rooted at $r$, for a vertex $v$ to verify that the root-to-$v$ path in $T_1$ and the root-to-$v$ path in $T_2$ are internally vertex-disjoint, we provide $v$ with the lists of its ancestors in two trees. Then, if $T_1$ and $T_2$ are not internally vertex-disjoint, then there exists a non-root vertex $v$, such that the paths from the root to $v$ of both trees go through another non-root vertex $u$. This means that $u$ is an ancestor of $v$ in both trees, and $v$ could detect $u$ in the lists of its ancestors in both trees. However, this approach would require $n\log n$ bits, which is too expensive. We show that the complexity can be reduced to $\Tilde{O}(\sqrt{n})$ via heavy-light decomposition.

In \Cref{sec: 2-indep-spt-cert-comp}, we will describe the certificate assignment and its complexity; in \Cref{sec: 2-indep-spt-vrfy-pf}, we will describe the verification algorithm and the correctness of the entire certification scheme. Then \Cref{lem: 2-indep-tree-certification} follows. For the rest of the proof, when $T$ is a rooted tree, we denote by $T(v)$ the subtree of $T$ rooted at $v$.

\subsubsection{Certification assignment and the complexity}\label{sec: 2-indep-spt-cert-comp}
\paragraph{Heavy-light decomposition and path-encoding}
\begin{definition}[{Heavy-light decomposition~\cite{SLEATOR1983362_heavy-light}}]
    The heavy-light decomposition of a rooted tree $T$ is a partition of its edges into heavy and light edges, such that each non-leaf vertex has one heavy edge to one of its children, which is the child with the largest number of vertices in its subtree.
\end{definition}
The heavy edges in $T$ induce a set of paths, called \emph{segments}. Each leaf in $T$ that is not adjacent to any heavy edge is treated as a segment of length $0$.
This way, each vertex belongs to exactly one segment.

\begin{figure}[ht]
    \centering
    \input{hldecomp}
    \caption{An illustration of the heavy-light decomposition. The segments are marked in black. For example, vertex $v$ has two children $v_1$ and $v_2$, and $v_1$ has more vertices in its subtree, and therefore belongs to the same segment as $v$. If a leaf does not belong to the same segment as its parent, it is itself considered a segment, see $v_3$.}
    \label{fig:heavy-light}
\end{figure}
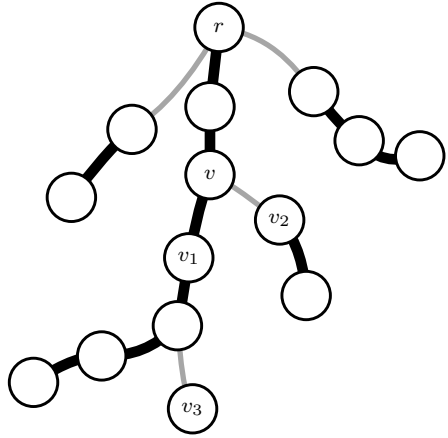

\begin{fact}\label{fact: heavy-light_fact_0}
    Consider any vertex in the tree $T$, the root-to-$v$ path can overlap with at most $O(\log n)$ segments.
\end{fact}

Let $T$ be a tree rooted at $r$. Denote by $\Psi(T)=\{S_1,\ldots,S_t\}$ the set of segments of the heavy-light decomposition on $T$. 

\begin{itemize}
    \item For each vertex $v\in T$, define $\seg_T(v)=i$ if $v\in S_i$. Recall that each vertex in $T$ belongs to exactly one segment. Moreover, we define the \emph{rank} of $v$ to be $j$ if $v$ is the $j$-th closest vertex to the root in $S_{\seg_T(v)}$, denoted as $\rank_T(v)=j$. We encode the position of $v$ in $T$ by $\eta_T(v)=(\seg_T(v),\rank_T(v))$.
    \item We encode the root-to-$v$ path, denoted by $p_T(v)$, as follows: Suppose that the root-to $v$ path in $T$ overlaps with $l$ segments $S'_1,\ldots,S'_l$. For each $j\in[l]$, let $u_j$ be the vertex in $S'_j\cap p_T(v)$ with the largest rank. We define the path-encoding of $v$, denoted by $\alpha_T(v)$, to be a set of segment-rank pairs, which consists of the vertex-encoding of $u_1,\ldots,u_l$. More formally, $\alpha_T(v)=\{(\seg_T(u_j), \rank_T(u_j)):j\in[l]\}$. Observe that the vertex-encoding of $v$ should be the last segment-rank pair contained in $\alpha_T(v)$, that is, $\eta_T(v)=(\seg_T(u_l), \rank_T(u_l))$. According to \Cref{fact: heavy-light_fact_0}, $l\in O(\log n)$, and therefore the path-encoding requires $O(\log^2n)$ bits.
\end{itemize}

For each vertex $v$, let $A_T(v)$ be the set of the addresses of all ancestors of $v$ in $T$. $A_T(v)=\{\eta_T(w):w\in p_T(v)\}$. Note that according to the definition of heavy-light decomposition, all segments must consist of consecutive vertices in an ancestor-descendant relation. Therefore, it holds for every segment that intersects with $p_T(v)$, the smallest-rank vertex in the intersection always has rank $1$. It follows that $A_T(v)$ can be recovered directly from the path-encoding $\alpha_T(v)$.

\paragraph{Certificate assignment}
Observe that to certify the independence of $T_1$ and $T_2$, it suffices for each non-root vertex $u$ to detect whether there exists another non-root vertex $v$ such that $u\in T_1(v)\cap T_2(v)$. For $i \in \{1,2\}$, let $-i$ be the complement of $i$ in the set $\{1,2\}$;  we say that $v$ is a \emph{$T_i$-ancestor} of $u$ if and only if $u\in T_i(v)$. For simplicity, let $\eta_i(v)=\eta_{T_i}(v)$ and $\alpha_i(v)=\alpha_{T_i}(v)$. We say that $(i,\alpha_i(v))$ is the \emph{$T_i$-address} of $v$.

\begin{itemize}
    \item For each vertex $v$, for all $i \in \{1,2\}$, include the $T_i$-address $(i,\alpha_i(v))$ to the certificate of $v$. 
    \item For each vertex $v$, we choose one of the following certificate assignments. And we will define how to choose between the two options later.
    \begin{itemize}
        \item (Option $1$.) Include $(2,\alpha_2(v))$ in the certificates of all the vertices in $T_1(v)$.
        \item (Option $2$.) Include $(1,\alpha_1(v))$ in the certificates of all the vertices in  $T_2(v)$.
    \end{itemize}
    In other words, for each $i \in \{1,2\}$, each vertex $u$ receives a list of $T_{-i}$-addresses of a subset of its $T_i$-ancestors (including itself). Denote this list by $\Lambda_i(u)$. Let the above choice of $v$ be $\pi(v)$. $\pi(v)=i$ means that we choose Option $i$ for $v$. We include this choice in the certificate of $v$ as well.
    \item (Cancellation rules.) For two distinct vertices $u,v$ such that $v\in T_1(u)$ and $u\in T_2(v)$, we remove some certificates based on the following rules:
    \begin{itemize}
        \item If $\pi(u)=\pi(v)=1$, then for every vertex $x\in T_1(v)$, we remove $(2,\alpha_2(u))$ from $\Lambda_1(x)$.
        \item If $\pi(u)=\pi(v)=2$, then for every vertex $x\in T_2(u)$, we remove $(1,\alpha_1(v))$ from $\Lambda_2(x)$.
    \end{itemize}
\end{itemize}

For each $i \in \{1,2\}$ and each segment $S\in\Psi(T_i)$, let $V(S)$ be the vertex set of this segment. Then, according to the cancellation rules, we have the following lemma:
\begin{lemma}\label{lem: seg-uniq}
    Given that $T_1$ and $T_2$ are independent. For each $i \in \{1,2\}$ and each segment $S\in\Psi(T_i)$, for any vertex $x\in G$, there exists at most one non-root vertex $v\in V(S)$ such that $(i,\alpha_i(v))\in\Lambda_{-i}(x)$.
\end{lemma}
\begin{proof}
    We assume that $i=2$ without loss of generality. Suppose that for some segment $S\in\Psi(T_2)$, there exist two non-root vertices $u,v\in V(S)$ and some vertex $x\in G$ such that $(2,\alpha_2(u)),(2,\alpha_2(v))\in\Lambda_1(x)$. According to the two options of sending information to children, we know that both $u$ and $v$ are ancestors of $x$ in $T_1$ (that is, $x\in T_1(u)\cap T_1(v)$) and that $\pi(u)=\pi(v)=1$. Assume that $v\in T_1(u)$ without loss of generality. Since $v$ and $u$ belong to the same segment in $T_2$, they must be in the ancestor-descendant relationship. Since $T_1$ and $T_2$ are independent and $u$ is a non-root ancestor of $v$ in $T_1$, it must be the case that $u\in T_2(v)$. According to the cancellation rules, we must remove $(2,\alpha_2(u))$ from $\Lambda_1(x)$. A contradiction.
\end{proof}

Given this lemma, we now describe the rules for deciding the choice $\pi(v)$. 
\begin{itemize}
    \item Let $\tau$ be a threshold that will be defined later. We say that a segment is \emph{long} if its length exceeds $\tau$; otherwise, it is \emph{short}. For any rooted tree $T$, let $\Long(T)$ be the set of vertices that belong to some long segments in $T$, and let $\Short(T)$ be the set of vertices that belong to some short segments in $T$.
    \item If $v\in\Long(T_1)$, then $\pi(v)=2$. Otherwise, $\pi(v)=1$.
\end{itemize}

\paragraph{Complexity}
Intuitively, when a vertex receives the addresses from a set of vertices, we can bound the number of vertices that come from long segments, since there are not many long segments globally. We can bound the number of vertices that come from short segments, since the total number of vertices in each short segment along the path from any vertex to the root is small. More precisely, the following holds:
\begin{itemize}
    \item If $(1,\alpha_1(v))$ occurs in the certificates of $u$, then $v\in\Long(T_1)$. According to \Cref{lem: seg-uniq}, the number of such $v$ is at most the number of long segments in $T_1$ plus one, which is $O(\frac{n}{\tau})$.
    \item If $(2,\alpha_2(v))$ occurs in the certificates of $v$, then $\pi(v)=1$ according to our rule, and one of the following holds:
    \begin{itemize}
        \item $v\in\Long(T_2)$. In this case, according to \Cref{lem: seg-uniq}, the number of such $v$ is at most the number of long segments in $T_2$ plus one, which is $O(\frac{n}{\tau})$.
        \item Observe that if $v\in\Short(T_2)$ and $\pi(v)=1$, then $v\in\Short(T_1)$ according to our rules. Therefore, $v\in\Short(T_1)\cap P_1(r,u)$. Recall that the number of short segments overlapping with the root-to-$u$ path in $T_1$ is at most $O(\log n)$, and the number of such $v$ is at most $\tau\cdot O(\log n)$. 
    \end{itemize}
\end{itemize}
For each vertex $v$ and each $i \in \{1,2\}$, the set $\Lambda_i(v)$ consists of $O(\frac{n}{\tau}+\tau\cdot O(\log n))$ addresses, each of which need $O(\log^2 n)$ bits to store. Set $\tau=\sqrt{\frac{n}{\log n}}$, we get an $O(\sqrt{n}\log^{2.5}n)$ size bound for our local certification scheme.

    

\subsubsection{Verification algorithm and correctness}\label{sec: 2-indep-spt-vrfy-pf}
\paragraph{Verification algorithm}
Recall that in our certificate assignment, for each $i \in \{1,2\}$, each vertex $v$ receives its own address $(i,\alpha_i(v))$, the choice $\pi(v)$ of which subtree to send its information to, and a list $\Lambda_i(v)$ of the $T_{-i}$-addresses of a set of its $T_i$-ancestors. We define $|\alpha_i(v)|$ to be the number of segment-rank pairs included in $\alpha_i(v)$, and we set $(\seg_i(v),\rank_i(v))$ to be the last segment-rank pair in $\alpha_i(v)$

The verification algorithm is as follows:
\begin{itemize}
    \item The vertices check the consistency between their path-encoding. More precisely, for each $i\in\{1,2\}$, each non-leaf vertex $v$ checks the following:
    \begin{itemize}
        \item (Exactly one child that belongs to the same segment.) There exists exactly one child $u$ satisfying $|\alpha_i(u)|=|\alpha_i(v)|$. Moreover, $\alpha_i(u)$ and $\alpha_i(v)$ differ only in the last segment-rank pair, which satisfies $\seg_i(u)=\seg_i(v)$ and $\rank_i(u)=\rank_i(v)+1$.
        \item (Other children initiate new segments.) For every child $w$ other than $u$, $|\alpha_i(w)|=|\alpha_i(v)|+1$. Moreover, $\alpha_i(w)=\alpha_i(v)\cup\{(\seg_i(w),\rank_i(w))\}$, and $\rank_i(w)=1$.
        

    \end{itemize}
    \item Each vertex checks that the certificate assignments respect our rules. More specifically, for each $i \in \{1,2\}$, each non-leaf vertex $v$ checks the following:
    \begin{itemize}
        \item (The two options.) Each vertex checks that either $\pi(v)=2$ and $\alpha_1(v)\in\Lambda_2(v)$, or $\pi(v)=1$ and $\alpha_2(v)\in\Lambda_1(v)$.
        \item (The cancellation rules.) For each child $x\in T_i(v)$, the lists of the addresses $\Lambda_i(v)$ and $\Lambda_i(x)$ satisfy the following: ($1.$) $x$'s $T_{-i}$-address does not appear in $\Lambda_i(v)$, that is, $\alpha_{-i}(x)\not\in\Lambda_i(v)$; and ($2.$) if there exists $u$ such that $\alpha_{-i}(u)\in\Lambda_i(v)\setminus\Lambda_i(x)$, then $x$ must be a $T_{-i}$-ancestor of $u$ and $\pi(x)=i$. This can be checked by the path encoding in $\alpha_{-i}(u)$ and the $T_{-i}$-address of $x$.
    \end{itemize}
    If any check fails, vertex $v$ outputs \reject.
    \item For each vertex $v$ and each $i \in \{1,2\}$, if there exists a label $(i,\alpha_i(u))$ in the certificate for some $u\neq v$ such that $\eta_i(u)\in A_i(v)$, then output \reject.
    \item The vertex outputs \accept\ if all the above verification passes.
\end{itemize}

\paragraph{Correctness}
Suppose that $T_1$ and $T_2$ are not independent, then there exists vertices $u,v$ such that $v\in T_1(u)\cap T_2(u)$. We assume $\pi(u)=1$ without loss of generality. Let $w$ be the vertex with the largest rank in $T_1$ that satisfies $\pi(w)=1$, $u\in T_2(w)$ and $v\in T_1(w)$. Such a vertex $w$ exists since we can take $w=u$. $w\neq v$ since otherwise $v\in T_1(u)\cap T_2(u)$ and $u\in T_2(w)$ will give us a contradiction to the fact that $T_{-i}$ is a tree. 

We claim that $(2,\alpha_2(w))$ must be added to the set $\Lambda_1(v)$. Suppose not, since $\pi(w)=1$, $(2,\alpha_2(w))\in\Lambda_1(w)$, we can find two vertices $y,z$ along the $w$--$v$ path in $T_1$ such that $y$ is the parent of $z$ and $(2,\alpha_2(w))\in\Lambda_1(y)\setminus\Lambda_1(z)$. According to our verification algorithm, this means that $z$ must be an $T_2$-ancestor of $w$ and $\pi(z)=i$, otherwise $y$ will output \reject. However, $z$ has a larger rank than $w$ in $T_i$, contradicting the way we chose $w$.

Since $v\in T_1(u)\cap T_1(w)$, $v\in T_2(u)$ and $u\in T_2(w)$, we have $v\in T_1(w)\cap T_2(w)$. According to our path encoding, $w$ belongs to $A_2(v)$, and $v$ will detect that $\eta_2(w)\in A_2(v)$ with the certificates $\alpha_2(w)$ and $\alpha_2(v)$, which make $v$ output \reject\ according to our verification algorithm.





\subsection{Lower bound for certifying two independent spanning trees}\label{sec:lower-2indst}
Since the upper bound for $k$-vertex-connectivity is $\Tilde{O}_k(\sqrt{n})$, while the lower bound we get in~\Cref{sec: lower} is $\Omega_k(\log n)$, a natural question to ask is whether we can do better using the independent tree approach. In this section, we show that certifying the independence of two prescribed spanning trees rooted at the same vertex requires $\Omega(\sqrt{n})$ bits. More formally, each vertex in the network receives its parents' IDs in both spanning trees as input, and the task is to verify that the two spanning trees are independent. We construct the lower bound via a standard reduction from the two-party communication complexity of the Set-Disjointness problem. 

\thmIndepST*

In~\Cref{sec: 2partyComm}, we will discuss the background about two-party communication complexity, and in~\Cref{sec: IST-lb-construction}, we will give the proof of the lower bound.

\subsubsection{Two-party communication complexity}\label{sec: 2partyComm}
\paragraph{The Set-Disjointness problem}
Let $x,y\in\{0,1\}^s$ be two $s$-bit strings, each of which can be viewed as an indicator vector of a set. We say that $x$ and $y$ are not disjoint if and only if there exists an index $i\in[s]$ such that $x_i=y_i=1$. The truth value of whether $x$ and $y$ are disjoint is denoted as $\DISJ(x,y)$. The two-party communication problem of Set-Disjointness ($\DISJ$) is defined as follows. Two players called Alice and Bob, are given inputs $x,y\in\{0,1\}^s$ respectively, and they communicate via a deterministic protocol to decide whether $\DISJ(x,y)=\true$ or $\DISJ(x,y)=\false$. The cost of a given protocol is the worst-case number of bits that the players need to communicate when using this protocol. The communication complexity of $\DISJ$ is the minimum cost among all the protocols.

\paragraph{Nondeterministic communication complexity}
In a nondeterministic protocol, each player receives, in addition to its input, an advising string provided by a third party that sees both the inputs of Alice and Bob. We say that the nondeterministic protocol solves some Boolean function $f$ if the following holds:
\begin{itemize}
    \item For each $\true$-instance, there exists an advising string such that the players output \accept.
    \item For each $\false$-instance, for any advising string, some player outputs \reject.
\end{itemize}
The cost of a nondeterministic protocol is the worst-case length of the advising string. The nondeterministic communication complexity of $f$ is the minimum cost among all the nondeterministic protocols that solve $f$.

For example, the nondeterministic communication complexity of $\overline{\DISJ}$ is $O(\log s)$ as we can pick the advising string to be the index $i$ where $x[i]=y[i]=1$. Each player can check that the $i$-th bit of their input is indeed $1$.

In contrast, it is known that the nondeterministic complexity of $\DISJ$ is $\Omega(s)$.
\begin{fact}[{\cite{Kushilevitz_Nisan_1996}}]\label{fact: DISJ-com}
    The nondeterministic communication complexity of $\DISJ(x,y)$ is $\Omega(s)$.
\end{fact}


\subsubsection{Lower-bound construction}\label{sec: IST-lb-construction}

Let $s$ be an integer. For every instance $(x,y)$ of a Set-Disjointness problem over $s^2$ elements, we define an instance $G_{x,y}$ on $n=(2(s+1)s+5)$ vertices. See \Cref{fig: 2-indep-st-lower} for an illustration.

\begin{itemize}
    \item The underlying graph is constructed as follows:
    \begin{itemize}
        \item The vertex set $V(G)=A\cup B\cup P\cup Q\cup\{v_1,v_2,u_1,u_2,r\}$, where $A=\{a_{i,j} : i,j\in[s]\}$, $B=\{b_{i,j} : i,j\in[s]\}$, $P=\{p_j:j\in[s]\}$, and $Q=\{q_i:i\in[s]\}$.
        \item $A$ forms a grid-like structure where each column is a clique of size $s$: For each $j\in[s]$, let $A_j=\{a_{i,j}:i\in[s]\}$, then $A_j$ induces a $K_s$. For each $i\in[s]$, $\{a_{i,j},a_{i,j+1}\}\in E(G)$ for all $j\in[s-1]$.
        \item $B$ forms a grid-like structure where each row is a clique of size $s$: For each $i\in[s]$, let $B_i=\{b_{i,j}:j\in[s]\}$, then $B_i$ induces a $K_s$. For each $j\in[s]$, $\{b_{i,j},b_{i+1,j}\}\in E(G)$ for all $i\in[s-1]$.
        \item For each $j\in[s]$, the vertex $p_j$ is connected to each vertex in $A_j\cup\{b_{1,j}\}$ with an edge. For each $i\in[s]$, the vertex $q_i$ is connected to each vertex in $B_i\cup\{a_{i,s}\}$ with an edge.
        \item The vertex $v_1$ is connected to each vertex in $P\cup\{v_2,r\}$ with an edge, and $v_2$ is connected to each vertex in $A\cup\{v_1,r\}$ with an edge. The vertex $u_1$ is connected to each vertex in $B\cup Q\cup\{u_2,r\}$, and $u_2$ is connected to each vertex in $Q\cup B_s\cup\{u_1,r\}$.
    \end{itemize}
    \item The two spanning trees on the graph $G_{x,y}$ are constructed as follows:
    \begin{itemize}
        \item For each $i\in[s]$, $(a_{i,j},a_{i,j+1})\in T_1$ for all $i\in[s-1]$ and $(a_{i,s},q_i)\in T_1$. For each $j\in[s]$, $(p_j,v_1)\in T_1$. Also, $(v_1,r)\in T_1$ and $(u_1,r)\in T_1$.
        \item For each $j\in[s]$, $(b_{i,j},b_{i+1,j})\in T_2$ for all $i\in [s-1]$ and $(p_j,b_{1,j}),(b_{s,j},u_2)\in T_2$. For each $i\in[s]$, $(q_i,u_2)\in T_2$. Also, $(v_2,r)\in T_2$ and $(u_2,r)\in T_2$.
        \item For each $i,j\in[s]$, if $x_{i,j}=0$, then $(a_{i,j},v_2)\in T_2$. For each $j\in[s]$, let $A'_j=A_j\cap\{a_{i,j}:x_{i,j}=1\}$. Let $|A'_j|=\nu_j$, for each $t\in[\nu_j]$, denote by $A'[t]$ the $t$-th element in $A'_j$ sorted by the $i$-index. Then $(A'_j[t],A'_j[t-1])\in T_2$ for all $t\in\{2,\ldots,\nu_j\}$ and $(A'_j[1],p_j)\in T_2$.
        \item For each $i,j\in[s]$, if $y_{i,j}=0$, then $(b_{i,j},u_1)\in T_1$. For each $i\in[s]$, let $B'_i=B_i\cap\{b_{i,j}:y_{i,j}=1\}$. Let $|B'_i|=\mu_i$, for each $t\in[\mu_i]$, denote by $B'_i[t]$ the $t$-th element in $B'_i$ sorted by the $j$-index. Then $(q_i,B'_i[1])\in T_1$, $(B'_i[t-1],B'_i[t])\in T_1$ for all $t\in\{2,\ldots,\mu_i\}$, and $(B'_i[\mu_i],u_1)\in T_1$.
    \end{itemize}
\end{itemize}

\begin{figure}
    \centering
    \includegraphics[width=\textwidth]{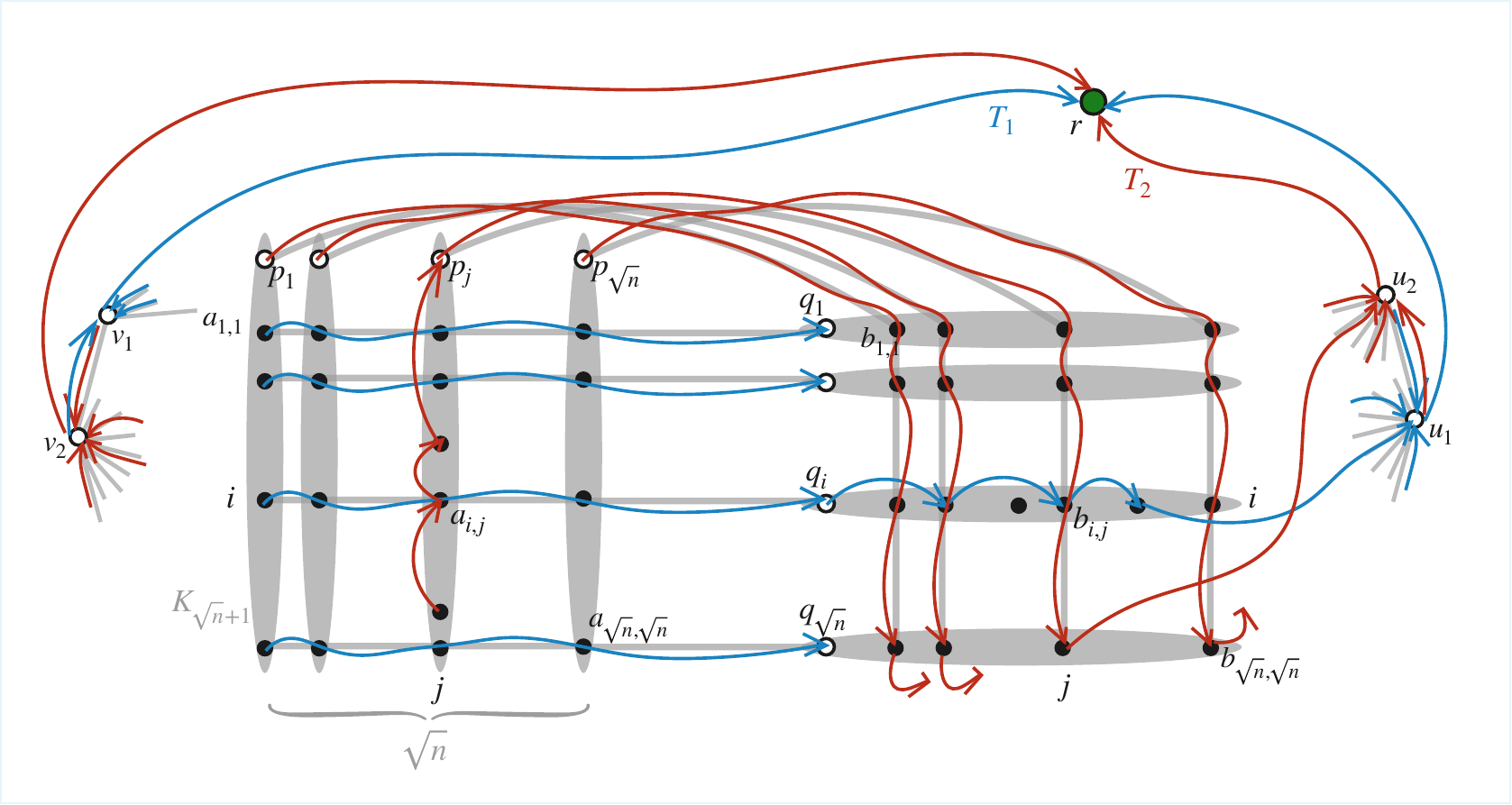}
    \caption{Lower bound construction for certifying two independent spanning trees.}
    \label{fig: 2-indep-st-lower}
\end{figure}

\begin{lemma}\label{lem: DISJ-to-indep-st}
    $T_1$ and $T_2$ are independent if and only if $\DISJ(x,y)=\true$.
\end{lemma}
\begin{proof}
    ($\Rightarrow$) If $\DISJ(x,y)=\false$, then there exists an index $(i,j)$ such that $x_{i,j}=y_{i,j}$. According to our construction of $G_{x,y}$, there exists a path from $a_{i,j}$ to $p_{j}$ and then to $b_{i,j}$ in $T_2$; on the other hand, there is also a path from $a_{i,j}$ to $q_i$ to $b_{i,j}$ in $T_1$. Therefore, the root-to-$a_{i,j}$ path in $T_1$ and $T_2$ intersect at $b_{i,j}$, which implies that $T_1$ and $T_2$ are not independent.

    ($\Leftarrow$) We prove that if $T_1$ and $T_2$ are not independent, then $\DISJ(x,y)=\false$.
    If there exists two distinct non-root vertices $\alpha$ and $\beta$ such that the root-to-$\alpha$ paths in $T_1$ and $T_2$ intersect at $\beta$, then $\beta$ must have one outgoing edge in $T_1$, one outgoing edge in $T_2$, and at least one incoming edge in both $T_1$ and $T_2$ according to our construction of $G_{x,y}$, which implies that $\beta\in A\cup B$. Observe that $\alpha$ should also belong to $A\cup B$. (If $\alpha\in P\cup Q\cup\{v_1,v_2,u_1,u_2,r\}$, then the only possible choice of $\beta$ is $r$, which is not allowed.) It is not possible that both $\alpha,\beta$ belong to $A$: Suppose that $\alpha=a_{i,j}$, then from the $T_1$ perspective, $\beta$ must have a same $i$-index as $\alpha$; from the $T_2$ perspective, $\beta$ must have a same $j$-index as $\alpha$, this would imply $\alpha=\beta$, a contradiction. Similarly, it is impossible that both $\alpha,\beta\in B$. Therefore, we know that $\alpha\in A$ and $\beta\in B$. Suppose that $\alpha=a_{i,j}$ and $\beta=b_{i',j'}$, then $\beta$ is in the root-to-$\alpha$ path in $T_1$ if and only if $i'=i$ and there exists a path from $q_i$ to $\beta$ in $T_1$, which implies that $y_{i',j'}=1$. On the other hand, $\beta$ is in the root-to-$\alpha$ path in $T_2$ if and only if $j'=j$, and there exists a path from $a_{i,j}$ to $p_j$, which implies that $x_{i,j}=1$. In conclusion, we must have $x_{i,j}=1=y_{i',j'}=y_{i,j}$ since $i'=i$ and $j'=j$. 
\end{proof}

\thmIndepST*

\begin{proof}
    Consider any local certification scheme for certifying two independent spanning trees. We construct a nondeterministic protocol for $\DISJ(x,y)$ as follows:
    Alice and Bob are given the input $x$ and $y$, respectively. They then construct $G_{x,y}$ based on their inputs. Also, they are both given an advising string that consists of all the certificates along the cut $(A\cup P\cup\{v_1,v_2\},B\cup Q\cup\{u_1,u_2,r\})$ in $G_{x,y}$. Each player then enumerates all the possible labelings of the vertices that it simulates. If there exists a labeling such that all vertices in its part output \accept, then it outputs \accept.

    According to \Cref{lem: DISJ-to-indep-st}, the protocol computes the output $\true$ of $\DISJ$ correctly -- if $\DISJ(x,y)=\true$, then there exists an advising string along the cut that both Alice and Bob can fill in their part and \accept; on the other hand, if $\DISJ(x,y)=\false$, then there is no such advising string. By \Cref{fact: DISJ-com}, we know that the advising string must have length $\Omega(s^2)$. However, the number of vertices along the cut  $(A\cup P\cup\{v_1,v_2\},B\cup Q\cup\{u_1,u_2,r\})$ is $(4s+3)$, therefore, the complexity of any local certification scheme for certifying two independent spanning trees is $\Omega(\frac{s^2}{4s+3})=\Omega(\sqrt{n})$.
\end{proof}    
    

%% file: hldecomp.tex
\begin{tikzpicture}[
    vertex/.style={
        circle,
        draw=black,
        fill=white,
        line width=1.1pt,
        minimum size=6.4mm,
        inner sep=0pt,
        outer sep=0pt
    },
    vlabel/.style={
        font=\scriptsize\itshape,
        inner sep=0pt,
        outer sep=0pt
    },
    edge/.style={
        draw=black,
        line width=4pt,
        line cap=round
    },
    guide/.style={
        draw=gray!70,
        line width=2pt,
        line cap=round
    }
]

\newcommand{\Vertex}[2]{%
    \node[vertex] at (#1) {};
    \node[vlabel] at (#1) {#2};
}

\coordinate (r)      at (3.30,5.80);
\coordinate (a)      at (3.18,4.75);
\coordinate (v)      at (3.18,3.85);

\coordinate (l1)     at (2.15,4.45);
\coordinate (l0)     at (1.35,3.55);

\coordinate (u1)     at (4.55,4.95);
\coordinate (u2)     at (5.15,4.30);
\coordinate (u3)     at (5.95,4.10);

\coordinate (v1)     at (2.90,2.75);
\coordinate (b)      at (2.75,1.85);
\coordinate (c)      at (1.75,1.45);
\coordinate (d)      at (0.85,1.10);

\coordinate (v3)     at (2.95,0.75);

\coordinate (v2)     at (4.10,3.25);
\coordinate (w)      at (4.45,2.25);

\draw[guide] (l1) .. controls (2.55,4.70) and (2.95,5.15) .. (r);
\draw[guide] (r)  .. controls (3.95,5.70) and (4.25,5.35) .. (u1);
\draw[guide] (v)  .. controls (3.55,3.70) and (3.85,3.40) .. (v2);
\draw[guide] (b)  .. controls (2.80,1.35) and (2.82,1.00) .. (v3);

\draw[edge] (l0) .. controls (1.55,3.80) and (1.85,4.20) .. (l1);

\draw[edge] (r) -- (a);
\draw[edge] (a) -- (v);

\draw[edge] (u1) .. controls (4.78,4.65) and (4.95,4.45) .. (u2);
\draw[edge] (u2) .. controls (5.35,4.05) and (5.62,4.00) .. (u3);

\draw[edge] (v)  .. controls (3.02,3.38) and (2.96,3.05) .. (v1);
\draw[edge] (v1) .. controls (2.84,2.35) and (2.80,2.08) .. (b);
\draw[edge] (b)  .. controls (2.42,1.55) and (2.18,1.45) .. (c);
\draw[edge] (c)  .. controls (1.42,1.46) and (1.08,1.25) .. (d);

\draw[edge] (v2) .. controls (4.35,2.95) and (4.42,2.62) .. (w);

\Vertex{r}{$r$}
\Vertex{a}{}
\Vertex{v}{$v$}

\Vertex{l1}{}
\Vertex{l0}{}

\Vertex{u1}{}
\Vertex{u2}{}
\Vertex{u3}{}

\Vertex{v1}{$v_1$}
\Vertex{b}{}
\Vertex{c}{}
\Vertex{d}{}
\Vertex{v3}{$v_3$}

\Vertex{v2}{$v_2$}
\Vertex{w}{}

\end{tikzpicture}

%% file: lower_v2.tex
\section{Lower bounds for $k$-connectivity for any $k\ge3$}\label{sec: lower}
In this section, we prove that for any $k\ge 3$, certifying $k$-vertex-connectivity and $k$-edge-connectivity both require $\Omega((\log n)/k)$ bits. 

Intuitively, we can find a $k$-(edge or vertex)-connected graph $H$ where there exists a special edge $\{u,v\}$ whose two endpoints are of bounded degree. We select $\poly(n)$ copies of the graph on distinct identifier sets. Suppose that there exists an $o(\log n)$ local certification, then there exist two instances $H_1$ and $H_2$ where the local views of the special edges $\{u_1,v_1\}$ and $\{u_2,v_2\}$ are the same. Then one can apply the cut-and-plug technique to create a new connected graph with doubled size from $H_1$ and $H_2$, where edges $\{u_1,v_1\}$ and $\{u_2,v_2\}$ are ``kicked out'' and $u_1$ is reconnected to $v_2$ while $u_2$ is reconnected to $v_1$. The graph constructed this way always has both vertex and edge connectivity $2$. Moreover, the local view of each vertex remains unchanged if one does not consider the IDs of its neighbors; we handle the neighbor IDs as follows. In~\cite{GS16}, they use a slightly more complicated cut-and-plug technique and prove several $\Omega(\log n)$ lower bounds even when the neighbors' IDs are taken into account. In particular, they show that one can ``glue'' two cycles where all vertices output \accept to get one long cycle that does not satisfy the property where all vertices still output \accept. We adapt their technique to prove the lower bound for certifying connectivity where the local view of each vertex contains the IDs of its neighbors. 

The above idea applies to a wider class of properties, which we formalized in the following lemma. See \Cref{fig: crossing-tech} for a visualization. To describe and prove this lemma, we need some more notations.
A graph property $\Pcal$ is a set of graphs (recall that the property we consider is closed under isomorphism). Let $G=(V,E)$ be a graph, for each vertex $v\in G$, denote by $N(v)$ the neighbors of $v$ in $G$ (not including $v$). We use $\ID(v)$ to represent the identifier of $v$.


\begin{figure}
    \centering
    \begin{subfigure}{0.3\textwidth}
        \centering
        \includegraphics[width=\textwidth]{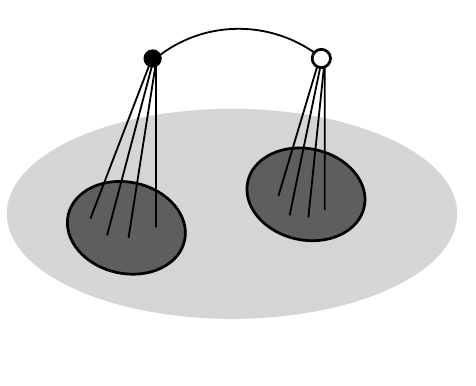}
        \caption{A graph with two special vertices (marked black and white) that are connected with an edge. Both special vertices have a constant number of neighbors.}
        \label{fig: crossing-1}
    \end{subfigure}
    \hfill
    \begin{subfigure}{0.6\textwidth}
        \centering
        \includegraphics[width=\textwidth]{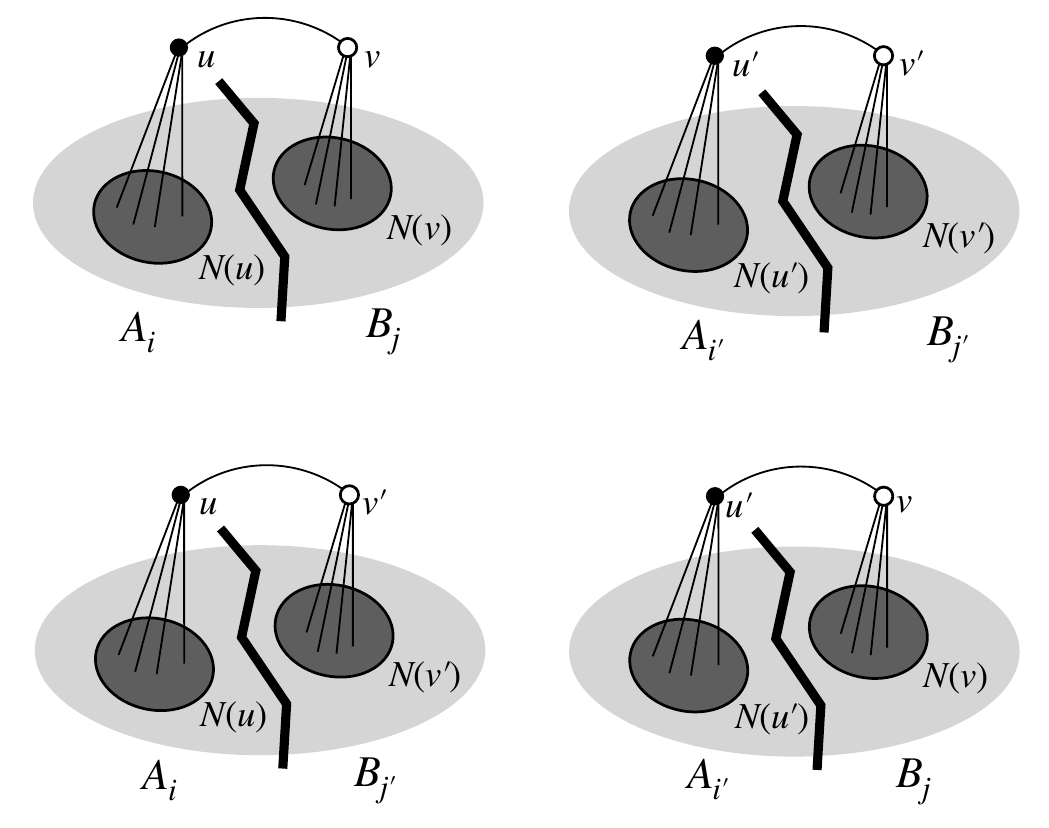}
        \caption{Four copies of the graph in~\Cref{fig: crossing-1} in different sets of identifiers $A_i\cup B_j$, $A_{i'}\cup B_{j'}$, $A_i\cup B_{j'}$, and $A_{i'}\cup B_j$. The copies of the black vertex and its neighbors are always on the $A$'s side, and the copies of the white vertex and its neighbors are always on the $B$'s side.}
        \label{fig: crossing-2}
    \end{subfigure}
    \begin{subfigure}{\textwidth}
        \centering
        \includegraphics[width=0.6\textwidth]{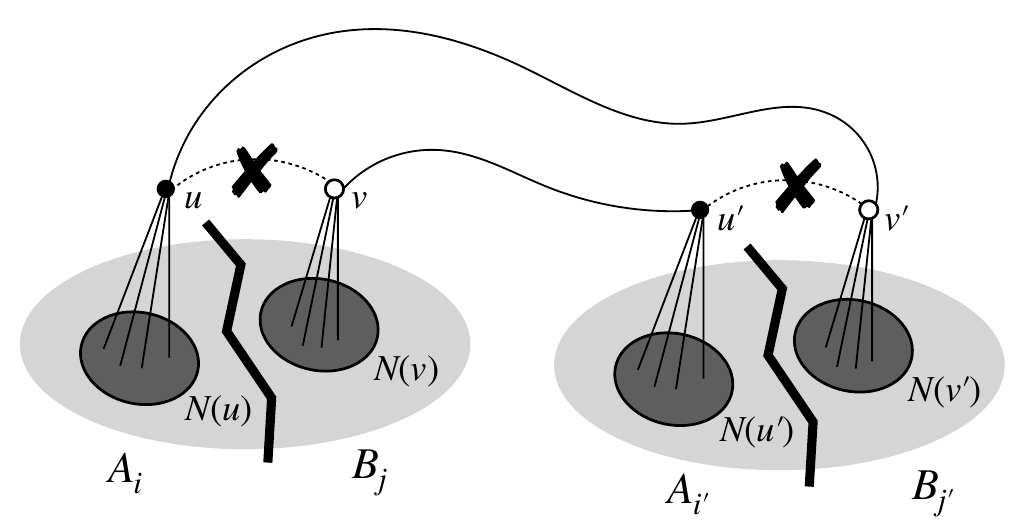}
        \caption{Take the copy with identifier set $A_i\cup B_j$ and the copy with identifier set $A_{i'}\cup B_{j'}$. Do a ``swap'' on the two edges connecting the special vertices in each graph. The local view of every vertex in the resulting graph should be the same as one of the copies in~\Cref{fig: crossing-2}. Therefore, the vertices still output accept, despite the fact that the graph is now $2$-connected.}
        \label{fig: crossing-3}
    \end{subfigure}
    \caption{An illustration of the general $\Omega(\log n)$ lower bound construction.}
    \label{fig: crossing-tech}
\end{figure}

\begin{lemma}\label{lem: general-lower-bound}
    Let $\Gcal\subseteq\Pcal$ be a class of graphs, let $\beta$ be an integer, and let $n_0$ be a constant.
    If for each $n\ge n_0$, there exists a graph $G\in\Gcal$ such that $G$ is of size at least $n$ and satisfies the following properties, then it requires $\Omega((\log n)/\beta)$ bits to certify $\Pcal$.
    \begin{itemize}
        \item [($i.$)] There exists an edge $\{u,v\}\in E(G)$ such that $|N(u)|,|N(v)|\leq\beta$, and $N(u)\cap N(v)=\emptyset$.
        \item [$(ii.$)] For the graph $G'$ built as follows, $G'\notin\Pcal$: Let $G_1$ and $G_2$ be two disjoint copies of $G$, and let $\{u_1,v_1\}\in E(G_1)$ and $\{u_2,v_2\}\in E(G_2)$ be the edges corresponding to $\{u,v\}$ in $G$. Define $G'$ by starting from the vertex-disjoint union $G_1 \uplus G_2$ and performing an edge swap: delete the edges $\{u_1,v_1\}$ and $\{u_2,v_2\}$, and add the edges $\{u_1,v_2\}$ and $\{v_1,u_2\}$. 
    \end{itemize}
\end{lemma}
\begin{proof}
    Given graph $G$ and edge $\{u,v\}\in E(G)$ 
    that satisfy condition (i.) and (ii.).
    Since the property $\Pcal$ is closed under isomorphism, we assume that the $\ID(u)=1,\ID(v)=n$, the IDs of $N(u)$ belong to $\{2,\ldots,|N(u)|+1\}$, and the IDs of $N(v)$ belong to $\{n-|N(v)|,\ldots,n-1\}$.

    Suppose the space of identifiers is $I=\{1, \ldots, n^{1+c}\}$ for some constant $c$. Let $N=n^c$. Consider any partition of $I$ into size-$(n/2)$ subsets $A_1, \ldots, A_N, B_1, \ldots, B_N$. For each identifier set $S$ defined above, we use $S[i]$ to denote the $i$-th identifier in $S$ in the natural order.

    For each pair of $(i,j)\in[N]^2$, we define a graph $C(i,j)$ that is isomorphic to $G$ (and therefore satisfies property $\Pcal$) on the identifier set $A_i\cup B_j$, such that for each $t=1,\ldots,\frac{n}{2}$, $A_i[t]$ is mapped to vertex $t$ in $G$ and $B_j[t]$ is mapped to vertex $\frac{n}{2}+t$ in $G$ under the isomorphism. For convenience, we denote $u_i=A_i[1]$ and $v_j=B_j[n]$.

    Let $c(i,j)$ be the ordered label sequence in the view of $u$ and $v$.
    \[c(i,j)=(A_i[|N(u)|+1],\ldots,A_i[1]=u_i,v_j=B_j[n],\ldots,B_j[n-|N(v)|]).\]
    Suppose that there exists four distinct indices $i,j,i',j'$ such that $c(i,j)=c(i',j')=c(i,j')=c(i',j)$, then we construct a graph $G'$ of size $n$ by gluing the graphs $C(i,j)$ and $C(i',j')$ together: We first remove the edge $(u_i,v_j)$ in $C(i,j)$ and the edge $(u_{i'},v_{j'})$ in $C(i',j')$, and then add two edges $(u_i,v_{j'})$ and $(u_{i'},v_j)$. We assign the certificates to vertices in $G'$ as in $C(i,j)$ and $C(i',j')$. One can make the following observations of $G'$:
    \begin{itemize}
        \item The local view of $u_i$ in $G'$ is the same as its local view in $C(i,j')$ since $c(i,j)=c(i,j')$. 
        \item The local view of $v_j$ in $G'$ is the same as its local view in $C(i',j)$ since $c(i,j)=c(i',j)$. 
        \item The local view of $u_{i'}$ in $G'$ is the same as its local view in $C(i',j)$ since $c(i',j')=c(i',j)$. 
        \item The local view of $v_{j'}$ in $G'$ is the same as its local view in $C(i,j')$ since $c(i',j')=c(i,j')$. 
        \item For all other vertices in $G'$, their local views remain unchanged.
    \end{itemize}
    Therefore, if there exists a certification scheme where all vertices in $C(i,j)$, $C(i,j')$, $C(i',j)$, and $C(i',j')$ output \accept, all the vertices in $G'$ should output \accept under our certification assignment, a contradiction. 

    It remains to show that such indices $(i,j,i',j')$ exist.

    Consider a complete bipartite graph $H=K_{N,N}$ on vertices $A \cup B$. Color the edge $(A[i],B[j])$ with $c(i,j)$. Suppose that we have a $o(\frac{\log n}{\beta})$ certification scheme, the information contains in $c(i,j)$ is less than $\frac{1}{\lambda}\cdot(|N(u)|+|N(v)|+2)\cdot\frac{\log n}{\beta}\leq \frac{4\log n}{\lambda}$ for arbitrary constant $\lambda$ when $n$ sufficiently large. Pick $\lambda=\frac{12}{c}$, the number of different colors is less than $2^{\frac{c\log n}{3}}=N^{\frac{1}{3}}$. Therefore, by the pigeonhole theorem, we have a subset of edges $E'\subseteq E(H)$ of size at least $\frac{N^2}{N^{\frac{1}{3}}}=N^\frac{5}{3}$ with the same color. According to \cite{BS74}, a length-4 cycle can be found in $E'$, which is a monochromatic cycle $A[i], B[j], A[i'], B[j']$ in $H$ for some $i,j,i',j'$, then $c(i,j)=c(i',j')=c(i,j')=c(i',j)$.
\end{proof}


Now we apply \Cref{lem: general-lower-bound} to prove lower bounds for certifying $k$-connectivity.

\thmKConnLower*
\begin{proof}
    Observe that the graph $G'$ constructed in \cref{lem: general-lower-bound} has both vertex-connectivity and edge-connectivity at most $2$. Therefore, property \textit{(ii)} holds, since we assume $k \ge 3$. Therefore, our goal is to find the $k$-edge-connected (or $k$-vertex-connected) graph class $\Gcal\subseteq\Pcal$ that satisfies property \textit{($i.$)} in \Cref{lem: general-lower-bound}. Let $n_0$ be a constant greater than $2k$, for any $n\ge n_0$, we construct a graph $G(n)$ of size at least $n$ which is $k$-edge-connected (or $k$-vertex-connected) as follows: 
    \begin{itemize}
        \item Pick any $k$-edge-connected ($k$-vertex-connected, resp.) graph $H(n)$ of size at least $(n-2)$.
        \item Add two vertices $u,v\not\in H(n)$. Pick two disjoint sets of $(k-1)$ vertices $S_u$ and $S_v$ in $H$. Connect each vertex in $S_u$ to $u$ with an edge, connect each vertex in $S_v$ to $v$ with an edge, and connect $u,v$ with an edge.
    \end{itemize}
    Note that for $k\ge 3$, we can choose $H(n)$ to be a complete graph of size $(n-2)$.
    According to our construction of $G(n)$, $|N(u)|=|N(v)|=k$. Therefore, we get the lower bound by applying \Cref{lem: general-lower-bound} to the graph family $\Gcal=\{G(n):n\ge 2k\}$.
\end{proof}
\begin{remark}
    There exist some choices of $H(n)$ such that $H(n)$ is of degree $O(k)$. Therefore, our lower bound works for bounded-degree graphs when $k$ is a constant. For example, one can consider a graph where we partition $n$ vertices into $\frac{n}{k}$ parts $V_1,\ldots,V_{\frac{n}{k}}$, each consisting of $k$ vertices. For each $i\in[\frac{n}{k}]$, $V_i$ and $V_{i+1}$ induce a complete bipartite graph where all edges connect a vertex in $V_i$ and a vertex in $V_{i+1}$.
\end{remark}


%% file: 2-conn_v2.tex
\section{Certifying $2$-connectivity in the sublogarithmic regime}\label{sec: 2-conn-bdd-deg}
In this section, we study the complexity of certifying $2$-connectivity in the sublogarithmic region.

\subsection{Constant-bit certification schemes for $2$-connectivity in restricted graph classes.}
In this section, we prove three constant-bit local certification schemes for $2$-edge-connectivity or $2$-vertex-connectivity stated in~\Cref{thm: 2-conn}.
\thmTwoConn*

We split~\Cref{thm: 2-conn} into three propositions and prove them separately in the following subsections. In~\Cref{subsec: 2-ec-bdd} (\Cref{thm: 2-edge-conn}), we give a constant-bit local certification scheme for $2$-edge-connectivity in bounded degree graphs. In~\Cref{subsec: 2-vc-bdd} (\Cref{thm: 2-vertex-conn}), we obtain a constant-bit local certification scheme for $2$-vertex-connectivity in bounded degree graphs via a reduction to $2$-edge-connectivity. Finally, in~\Cref{subsec: 2-ec-expansion} (\Cref{thm: 2-ec-bdd-expansion}), we extend the constant-bit certification scheme for $2$-edge-connectivity to bounded-expansion graph classes

\subsubsection{$2$-edge-connectivity in bounded-degree graphs}\label{subsec: 2-ec-bdd}
In this section, we show that $2$-edge-connectivity admits a constant-size local certification scheme on bounded degree graphs.
\begin{restatable}{proposition}{thmTwoEC}\label{thm: 2-edge-conn}
    There exists a local certification scheme for $2$-edge-connectivity on bounded-degree graphs with per-vertex label size $O(\Delta\log\Delta)$, even when the vertices are anonymous.
\end{restatable}

Our upper bound exploits the following lemma on Eulerian subgraph covering. 
A graph is Eulerian if and only if every vertex has even degree. Given a graph $G$, we say that an edge $e$ is a \emph{bridge} of $G$ if removing $e$ increases the number of connected components in $G$. A graph is \emph{bridgeless} if it does not contain a bridge. It is known that the edges of a bridgeless graph can be covered by three (possibly disconnected) Eulerian subgraphs.
\begin{lemma}[\cite{AT85}]\label{lem: Eulerian-cover}
    An undirected graph $G$ is $2$-edge-connected if and only if its edge set $E(G)$ can be covered by three (possibly disconnected) Eulerian subgraphs $(V(G),C_i)$ for $i \in [3]$; that is, $C_1 \cup C_2 \cup C_3 = E(G)$.
\end{lemma}
Our first observation is that $2$-edge-connectivity admits a constant-bit edge-certification scheme via a simple application of~\Cref{lem: Eulerian-cover}.

\begin{restatable}{lemma}{lemTwoECedge}\label{lem: 2-ec-EdgeCert}
    There exists an $O(1)$-bit edge-certification scheme for $2$-edge connectivity.
\end{restatable}

\begin{proof}
    For each $i\in[3]$, color all the edges in $C_i$ with color $i$. Each vertex then checks that ($1.$) for each color $i\in[3]$, the number of incident edges that are colored $i$ is even, and ($2.$) every edge is colored by at least a color in $[3]$.
\end{proof}

\Cref{thm: 2-edge-conn} then follows immediately by a transformation from edge-certification to vertex-certification.
\begin{proof}[Proof of~\Cref{thm: 2-edge-conn}]
    According to \Cref{lem: edge-cert-bdd-deg}, we can transform this $O(1)$-bit edge-certification scheme to an $O(\Delta\log\Delta)$-bit vertex certification scheme for graphs with max-degree $\Delta$.
\end{proof}


\subsubsection{$2$-vertex-connectivity in bounded degree graphs}\label{subsec: 2-vc-bdd}
In this section, we show that $2$-vertex-connectivity can be certified with $O(1)$ bits on bounded degree graphs via a reduction to $2$-edge-connectivity.
\begin{restatable}{proposition}{thmTwoVC}\label{thm: 2-vertex-conn}
    There exists a local certification scheme for $2$-vertex-connectivity on bounded-degree graphs with per-vertex label size $O(2^{\Delta}\cdot\Delta\log\Delta)$, even when the vertices are anonymous.
\end{restatable}

The intuition is to transform the graph $G$ into a set of graphs by making the vertices of $G$ into edges and reconnecting with its neighborhood in a way such that $G$ is $2$-vertex-connected if and only if all the transformed graphs are $2$-edge-connected.

\paragraph{Transforming a $2$-vertex-connected graph into a $2$-edge-connected graph}
Given a $2$-vertex-connected graph $G=(V,E)$, we construct a graph $G'=(V',E')$ by the following operation:
\begin{itemize}
    \item For each vertex $v\in G$, make it into a super vertex $\psi(v)=\{v_0,v_1\}$. Add the vertices $v_0,v_1$ to $V'$, and add an edge $\{v_0,v_1\}$ to $E'$.
    \item For each vertex $v\in G$, partition $N(v)$, the neighborhood of $v$, into two non-empty subsets $N_0(v),N_1(v)$.
    \item For each pair of vertices $\{u,v\}\in E$ and $i,j\in\{0,1\}$, add an edge $\{u_i,v_j\}$ to $E'$ if and only if $v\in N_i(u)$ and $u\in N_j(v)$.
\end{itemize}

According to the rules of how we construct $G'$, we have the following observation:
\begin{observation}\label{obs: 2-ve-to-2-ec}
    For each pair of vertices $u,v\in G$, $\{u,v\}\in E$ if and only if there exists some edge $\{u_i,v_j\}\in E'$ between $\psi(u)$ and $\psi(v)$.
\end{observation}

We claim that if the original graph $G$ is $2$-vertex-connected, then the resulting graph $G'$ is $2$-edge-connected.
\begin{claim}\label{claim: 2-conn-1}
    If $G$ is $2$-vertex-connected, then $G'$ is $2$-edge-connected.
\end{claim}
\begin{proof}
    Suppose that $G'$ is not $2$-edge-connected, then there exists a bridge $e=\{x,y\}\in E'$ that splits the graph into two connected components $X$ and $Y$ where $x\in X,y\in Y$. Suppose that $x\in \psi(u)$ and $y\in \psi(v)$ for some $u,v\in V$, then for each $w\in V$, $w\not\in\{u,v\}$ the two vertices in $\psi(w)\subseteq V'$ must belong to the same side of the cut $(X,Y)$. Let $S=\{s\in V:\psi(s)\subseteq X\}$ and $T=\{t\in V:\psi(t)\subseteq Y\}$. If $u\neq v$, then according to \Cref{obs: 2-ve-to-2-ec}, there is no edge between $S\cup\{u\}$ and $T\cup\{v\}$ except for $\{u,v\}$. Then removing either $u$ or $v$ will disconnect the graph. Otherwise, $u=v$ and $x=u_0,y=u_1$, then there is no edge between $S$ and $T$. Note that in this case, $S$ and $T$ will be non-empty since we have partitioned $N(u)$ into two non-empty sets. Thus, removing $u$ will disconnect the graph. In both cases, we arrive at a contradiction to the fact that $G$ is $2$-vertex-connected. 
\end{proof}

\begin{claim}\label{claim: 2-conn-2}
    If $G$ has a vertex cut $\{v\}$, then there exists a partition of $N(v)$ such that $\{v_0,v_1\}$ is a bridge in $G'$, no matter how we partition the neighborhoods of other vertices.
\end{claim}
\begin{proof}
    Let $S,T\subseteq V$ be two vertex sets that are disconnected after removing $v$ from $G$. Let $N_0(v)=N(v)\cap S$ and $N_1(v)=N(v)\cap T$. Then $\{v_0,v_1\}$ is a bridge in $G'$.
\end{proof}

\begin{proof}[Proof of \Cref{thm: 2-vertex-conn}]
    Let $G=(V,E)$ be a graph of maximum degree $\Delta$. Then for each $v\in V$, there exist at most $2^{\deg(v)-1}-1\leq 2^{\Delta-1}-1$ ways of partitioning $N(v)$. Let $t(v)=2^{\deg(v)-1}-1$, and $t=2^{\Delta-1}-1$. Let the set of partitions be $\Gamma(v)=\{\Gamma^i(v)\}_{i=1}^{t}$, where $\Gamma^i=\{N_0^i(v),N_1^i(v)\}$, $\Gamma^i(v)\neq\Gamma^j(v)$ for every $1\leq i<j\leq t(v)$ and $\Gamma^i(v)=\Gamma^{t(v)}(v)$ for all $t(v)\leq i\leq t$.

    For each $i\in[t]$, we construct a graph $G_i$ using the partition $\Gamma^i(v)$ for every vertex $v\in G$. According to \Cref{claim: 2-conn-1,claim: 2-conn-2}, if $G$ is $2$-vertex-connected, all $\{G_i\}_{i=1}^t$ are $2$-edge-connected. Otherwise, there exists some $G_{i^*},i^*\in[t]$ such that $G_{i^*}$ has a bridge. Therefore, it suffices to simulate the certification of $2$-edge-connectivity simultaneously for each $G_i$. Our certification scheme is as follows:
    \begin{itemize}
        \item For each $v\in V$, assign the list of partitions $\Gamma(v)=\{\Gamma^i(v)\}_{i=1}^t$ to $v$. Each of the partitions can be stored in $O(1)$ bits by first doing a $2$-hop coloring on $V$, and then encoding the partition of $N(v)$ by a partition of the colors in $N(v)$. 
        \item For each $i\in[t]$, let $c_i$ be the certificate assignment for $2$-edge-connectivity when the underlying graph is $G_i$. For each $v\in V$, assign $v$ the list of certificates $\{(c_i(v_0),c_i(v_1))\}_{i=1}^t$.
        \item Each vertex $v\in V$ runs the following verification algorithm:
        \begin{itemize}
            \item First check that the $2$-hop coloring (for encoding the partition) is proper, this can be done by having each vertex check that no color appears twice in their neighborhood. Then check that the set of partitions $\Gamma(v)$ indeed contains all possible partitions of its neighborhood, which can be done with the knowledge of its neighborhood in the original graph.
            \item For each $i\in[t]$, each vertex $v\in V$ has the knowledge of the local graph topology of $G_i$ by looking at the partition of its neighbors and itself. Therefore, it can simulate the verification algorithm of $v_0$ and $v_1$ on $G_i$ with the certificates $\{(c_i(v_0),c_i(v_1))\}$. If the verification fails for some $i$, it outputs \reject. Otherwise, it outputs \accept.
        \end{itemize}
    \end{itemize}
    The correctness of this certification scheme follows from the correctness of certifying $2$-edge-connectivity on bounded-degree graphs. The size of the certificates blows up by a factor of $O(t)=O(2^{\Delta})$, which gives us the size bound of certificates.
\end{proof}

\subsubsection{$2$-edge-connectivity for bounded-expansion graph classes}\label{subsec: 2-ec-expansion}
In this section, we show that the constant-bit certification for $2$-edge-connectivity in bounded degree graphs (\Cref{thm: 2-edge-conn}) can be further generalized to bounded-expansion graphs, which include several well-known graph classes such as bounded-degree graphs and minor-closed graph classes. The formal definition is stated below.

\begin{definition}[$t$-shallow minor]
    Let $G$ and $G'$ be two graphs. We say that $G'$ is a $t$-shallow minor of $G$ if there exists a subgraph $H$ of $G$ and a partition of the vertices of $H$ into connected sets with radius (measured in $H$) at most $t$, such that $G'$ is obtained from $H$ by contracting these connected sets.
\end{definition}

\begin{definition}[bounded expansion]
    A class of graphs $\Fcal$ is said to have bounded expansion if there exists a function $f$ such that for each graph $G\in\Fcal$, each $t\ge 1$, every $t$-shallow minor $G'$ of $G$ satisfies that $\frac{|E(G')|}{|V(G')|}\le f(t)$.
\end{definition}

\begin{restatable}{proposition}{thmTwoECnew}\label{thm: 2-ec-bdd-expansion}
    Let $\Fcal$ be a bounded-expansion graph class. Given the promise that the input graph belongs to $\Fcal$, there exists a local certification scheme for $2$-edge-connectivity with $O(1)$ per-vertex label size, where the constant only depends on the graph class $\Fcal$.
\end{restatable}

Before proving~\Cref{thm: 2-ec-bdd-expansion}, we first provide some intuition and review the necessary graph-theoretic background.

Our approach again builds on the edge certification scheme in~\Cref{lem: 2-ec-EdgeCert}. The main obstacle was that when the degree is unbounded, transforming edge-certification schemes to vertex-certification schemes requires vertices to know which edges it is simulating, which is precisely what gives rise to the additive $\Theta(\log n)$ term in~\Cref{lem:EdgeCert-degeneracy}. We show that for bounded expansion graphs, this obstacle can be overcome via an application of a very recent certification scheme for pseudoforest~\cite{BFVZ26}.

A \emph{pseudoforest} is a graph where each connected component has at most one cycle. It is easy to see that every pseudoforest admits an edge-orientation with maximum out-degree at most $1$. Given such an orientation, we say that a vertex $u$ is the parent of vertex $v$ if and only if $v$ is oriented into $u$; otherwise, $u$ is a child of $v$. Given a labeled graph $G$, we say that the labeling \emph{encodes} a pseudoforest in $G$ if for every vertex in $G$, by investigating the labels of its neighborhood, the vertex knows which neighbor is its parent and which neighbors are its children in that pseudoforest.

In~\cite{BFVZ26}, they show that for bounded expansion graphs, there exists an $O(1)$-bit encoding for a BFS spanning forest. Moreover, their encoding scheme guarantees that even if the prover lies, the certificates still at least encode a pseudoforest.

\begin{lemma}[{\cite{BFVZ26}}]\label{lem: psforest}
    Let $\Fcal$ be a bounded-expansion graph class. There exists an $O(1)$-bit local certification scheme that satisfies the following properties:
    \begin{itemize}
        \item (Completeness.) If the prover is honest, the certificate assignment encodes a BFS spanning forest.
        \item (Soundness.) If the verifier accepts on every vertex, then the certificate assignment encodes a pseudoforest.
    \end{itemize}
\end{lemma}

We then observe that, as long as a partition of edges into a constant number of pseudoforests is given as input to the vertices, transforming edge-certification to vertex-certification incurs only a constant overhead by letting the children take the responsibility of simulating the edge connected to the parent. More formally, a \emph{$d$-pseudoforest decomposition} of a graph is a partition of its edges into d subsets, each inducing a pseudoforest. Similarly, we say that a $d$-pseudoforest decomposition is encoded in a labeled graph if for each $i\in[d]$, each vertex can determine, by inspecting the labels in its neighborhood, which of its neighbors are its parent and children in the i-th pseudoforest. In the following lemma, we show that if a $d$-pseudoforest decomposition is encoded in the input graph, then we can transform an edge-certification scheme into a vertex-certification scheme with only a factor of $d$ overhead.

\begin{lemma}\label{lem: EdgeCert-oriented}
    Consider a graph family $\Fcal$ and an edge-certification for $\Fcal$ of size $f(n)$. If a $d$-pseudoforest decomposition is encoded for every graph in $\Fcal$, then there exists a (vertex) certification of size $d\cdot f(n)$.  
\end{lemma}
\begin{proof}
    Let $G=(V,E)$ be a graph in $\Fcal$. For each edge $e\in E$, denote by $c(e)$ the certificate assignment on $e$. Denote by $T_i$ the $i$-th pseudoforest in the decomposition. We now construct a vertex-certification scheme as follows:
    \begin{itemize}
        \item (Prover.) For each edge $e=\{u,v\}\in T_i$, suppose that $v$ is the parent of $u$ in $T_i$, remove the certificate $c(e)$ from $e$ and assign $(i,c(e))$ to the vertex $u$.
        \item (Verifier.) For each vertex $v\in V$, recover the edge certification for each incident edge by investigating the certificates of its neighbors, and then run the verification scheme as in the original edge certification scheme.
    \end{itemize}
    Since each vertex has at most $1$ parent in each pseudoforest, it receives at most $d$ certificates from its incident edges, and therefore the total certificate size is at most $d\cdot f(n)$.
\end{proof}

Combining with the edge-certification scheme for $2$-edge-connectivity in~\Cref{lem: 2-ec-EdgeCert}, we derive the following lemma.
\begin{lemma}\label{lem: 2-ec-oriented}
    Let $\Fcal$ be a labeled graph family such that a $d$-pseudoforest decomposition is encoded for every graph in $\Fcal$. Then there exists an $O(1)$-bit local certification scheme for $2$-edge-connectivity on $\Fcal$.
\end{lemma}
\begin{proof}
    First, apply~\Cref{lem: 2-ec-EdgeCert} to obtain an $O(1)$-bit edge-certification scheme and then transform it into a vertex-certification scheme using~\Cref{lem: EdgeCert-oriented} with a constant overhead.
\end{proof}

Our main idea is to use the pseudoforest certification scheme of~\Cref{lem: psforest} to simultaneously encode a sparse subgraph that preserves 2-edge-connectivity and a pseudoforest decomposition of this subgraph. We then apply the edge-certification scheme of~\Cref{lem: 2-ec-EdgeCert} to this subgraph and use the pseudoforest decomposition to redistribute the edge labels to the vertices, thereby obtaining a constant-bit vertex-certification scheme.

To extract this sparse core via~\Cref{lem: psforest}, we will invoke the following graph-theoretic tools. Let $G=(V,E)$ be a graph, a \emph{certificate} for $k$-edge-connectivity ($k$-vertex-connectivity, resp.) of $G$ is a subset of edges $E'\subseteq E$ such that the subgraph $G'=(V,E')$ is $k$-edge-connected ($k$-vertex-connectivity, resp.) if and only if $G$ is $k$-edge-connected ($k$-vertex-connectivity, resp.). It has been shown that for both $k$-edge-connectivity and $k$-vertex-connectivity, there exists a certificate of degeneracy at most $k$~\cite{NI92,CT90} by finding BFS (breadth-first search) spanning forest iteratively.

\begin{lemma}[\cite{NI92,CT90}]\label{lem: sparse-cert}
    Let $G=(V,E)$ be a graph, and for $i=1,\ldots,k$, let $F_i$ be the edge set of a BFS spanning forest of $G\setminus(\cup_{j=1}^{i-1}F_j)$. Then $\cup_{i=1}^k F_i$ is a certificate for both $k$-edge-connectivity and $k$-vertex-connectivity. 
\end{lemma}

Intuitively, letting $k=2$, we can use~\Cref{lem: psforest} to encode the sparse certificate in~\Cref{lem: sparse-cert} by encoding the BFS spanning forest iteratively. However, since provers may cheat, in which case we only get pseudoforests that do not span the entire graph, we need to extract another subset of edges to ensure that the bridge is preserved. 

\begin{lemma}[{\cite{JAEGER1979205}}]\label{lem: odd-matching}
    For any graph $G$ and any spanning tree $T$ of $G$, there exists a forest $F\subseteq T$ such that the graph $G'=G\setminus F$ obtained by removing every edge in $F$ from $G$ is an Eulerian subgraph of $G$.
\end{lemma}

Observe that every graph with a bridge is not an Eulerian graph, as there are at least two vertices with odd degree. Consequently, a subgraph that preserves bridges can be verified by a degree-parity check.

\begin{claim}\label{claim: bridge-preserving-subgraph}
    Given a graph $G=(V,E)$ that contains a bridge. For any edge set $E'\subseteq E$ satisfying that for each vertex $v\in V$, $\deg_{E'}(v)\equiv\deg_G(v)\mod 2$, $E'$ contains every bridge in $G$.
\end{claim}
\begin{proof}
    Let $G'$ be the subgraph of $G$ induced by the edge set $E-E'$. Observe that for each vertex $v\in V$, $\deg_{G'}(v)\equiv(\deg_G(v)-\deg_{E'}(v))\equiv0\mod 2$. Therefore, $G'$ is an Eulerian subgraph of $G$, and thus does not contain any bridge of $G$.
\end{proof}

Now we are ready to prove~\Cref{thm: 2-ec-bdd-expansion}.
\begin{proof}[Proof of~\Cref{thm: 2-ec-bdd-expansion}]~
\begin{description}
    \item[Certificate assignments.] Prover assigns the following certificates to the vertices:
    \begin{itemize}
        \item [1.] Specify two BFS spanning trees using~\Cref{lem: psforest}, with distinct sets of labels. Denoted as $T$ and $T_1$.
        \begin{itemize}
            \item [1-1.] Specify a forest $F\subseteq T$ whose removal makes the graph Eulerian. The existence of $F$ is guaranteed by~\Cref{lem: odd-matching}. To encode $F$, the prover does the following: For each edge $e\in F$, suppose that $e=(u,v)$ is oriented from $u$ to $v$, then the prover puts a special symbol $*$ on $u$.
            \item [1-2.] Specify a BFS spanning forest $T_2$ on the graph $G-T_1$ using~\Cref{lem: psforest} with a distinct set of labels from $T_1$.
        \end{itemize}
        Note that $T$ and $T_1$ can be arbitrary BFS spanning trees; they can also be identical. We distinguish them here only for the ease of presentation.
        \item [2.] Certify $2$-edge-connectivity on the subgraph $G'$ induced by the edge set $F\cup T_1\cup T_2$. Since we have certified three pseudoforests in step 1., we can apply the $O(1)$-bit certification scheme in~\Cref{lem: 2-ec-oriented}.
    \end{itemize}
    \item[Verification algorithm.] The verifier recovers the edge set encoded in item 1-1 by investigating the certificates assigned in its neighborhood. This edge set is denoted as $F'$ to distinguish between the set $F$ guaranteed in~\Cref{lem: odd-matching}.
    \begin{itemize}
        \item (Degree parity check.) Each vertex $v$ checks that $\deg_{F'}(v)\equiv\deg_G(v)\mod 2$.
        \item Each vertex runs the verification algorithm of $2$ on the subgraph induced by the edge set $F'\cup T_1\cup T_2$.
    \end{itemize}
\end{description}
Let us now prove the correctness of the algorithm. For completeness, we note that according to~\Cref{lem: sparse-cert}, the subgraph of $G$ induced by the edge set $T_1\cup T_2$ is a certificate for $2$-edge-connectivity. $G[F\cup T_1\cup T_2]$ is a subgraph of $G$ that contains $G[T_1\cup T_2]$, and therefore also a certificate for $2$-edge-connectivity of $G$. Since prover has encoded three pseudoforests $F$, $T_1$, and $T_2$, the completeness follows from applying~\Cref{lem: 2-ec-oriented} to $G[F\cup T_1\cup T_2]$.

For soundness, according to~\Cref{claim: bridge-preserving-subgraph}, if the degree parity check passes on every vertex, then $F'$ should contain all the bridges in $G$. Moreover, since prover has encoded three pseudoforests $F'$, $T_1$, and $T_2$, applying the soundness of the certification scheme in~\Cref{lem: 2-ec-oriented} to $G[F'\cup T_1\cup T_2]$, the verifier will reject when $F'$ contains a bridge.

\end{proof}

\begin{remark}
    The constant in~\Cref{thm: 2-ec-bdd-expansion} comes from the use of conflict-free colorings~\cite{BFVZ26}. Since a 2-hop coloring is a conflict-free coloring on bounded-degree graphs,~\Cref{thm: 2-edge-conn} is subsumed by~\Cref{thm: 2-ec-bdd-expansion}.
\end{remark}

\subsection{An $\omega(1)$ lower bound for $2$-vertex-connectivity in the anonymous model}
In contrast to the constant-size certification scheme for bounded-degree graphs in~\Cref{thm: 2-vertex-conn}, we give an $O(\log(\log^*n))$-bit lower bound for certifying $2$-vertex-connectivity in the anonymous model. The lower bound construction is a slight modification of the parity lower bound in~\cite{BFVZ26} with the observation that their accepting instance can be made $2$-vertex-connected. See~\Cref{fig: 2VC-lb} for an illustration.

\thmTwoVCLB*

To prove~\Cref{thm: 2-vc-lb}, we use the same approach as in~\cite{BFVZ26}. The following lemma is an implication of the Finite Union Theorem (see~\cite[Theorem 12, Theorem 13]{BFVZ26} and the reference therein) and Lemma 15 in~\cite{BFVZ26}. 
\begin{lemma}[{~\cite{graham1991ramsey, setyawan1998combinatorial, BFVZ26}}]\label{lem: FiniteUnionThm}
    For every integer $m$, let $c=f(m)$ be the minimum number of colors for which there exists a coloring of all non-empty subsets with $c$ colors such that no two disjoint subsets $\emptyset\subsetneq S_1,S_2\subseteq[m]$ satisfy that $S_1,S_2$, and $S_1\cup S_2$ all receive the same color. Then $c\in\Omega(\log^*m)$.
\end{lemma}

\begin{proof}[Proof of~\Cref{thm: 2-vc-lb}]
    We first construct a graph $H=(V,E)$ as follows:
    \begin{itemize}
        \item First, we have a ground set $M=[m]$. Let $B=M\cup\{i+m:i\in M\}=\{1,\ldots,m, m+1,\ldots,2m\}\subseteq V$. The set $B$ induces a clique $K_{2m}$ in $G$. That is, every vertex in $B$ is connected to all other vertices in $B$. 
        \item For each $S\subseteq M$, let $S_{\parallel m}=\{s+m:s\in S\}$ we introduce a new vertex $v_S\in V$ that connects to all vertices in $S\cup S_{\parallel m}$. Note that for every $S$, $(S\cup S_{\parallel m})\subseteq B$.
    \end{itemize}
    Observe that $H$ is indeed $2$-vertex-connected. According to~\Cref{lem: FiniteUnionThm}, if the number of labels is $o(\log^*m)$, then there exist two disjoint subsets $S_1$ and $S_2$ in $[m]$ such that $v_{S_1}$, $v_{S_2}$, and $v_{S_1\cup S_2}$ receive the same labels.

    Take two copies of $H$ with the same set of accepting certificate assignments. Denote the two copies by $H_1$ and $H_2$, respectively. We construct the graph $G$ by identifying the vertices $v_{S_1}$ in $H_1$ and $v_{S_2}$ in $H_2$, call this new vertex $v^*$. Then $v^*$ is a cut vertex in $G$. Observe that $v^*$ has the same label as $v_{S_1}$ and $v_{S_2}$, and therefore identifying these two vertices will not change the view of every other vertex than $v^*$. The view of $v^*$ is identical to the view of $v_{S_1\cup S_2}$ in both $H_1$ and $H_2$; therefore, all vertices still output \accept in $G$, a contradiction.
\end{proof}

\begin{figure}[ht]
    \centering
    \begin{subfigure}{0.4\textwidth}
        \centering
        \includegraphics[width=0.5\textwidth]{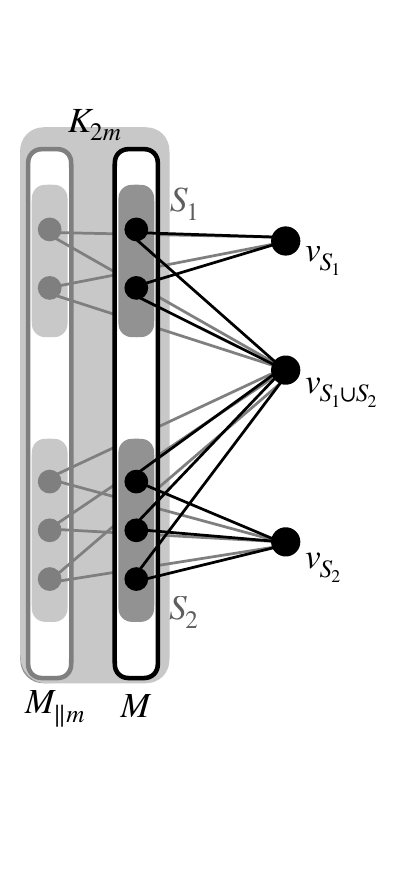}
        \caption{A $2$-vertex-connected graph. The vertices $v_{S_1}$, $v_{S_2}$, and $v_{S_1\cup S_2}$ receive the same certificates.}
        \label{fig: 2VClb-half}
    \end{subfigure}
    \hfill
    \begin{subfigure}{0.45\textwidth}
        \centering
        \includegraphics[width=0.7\textwidth]{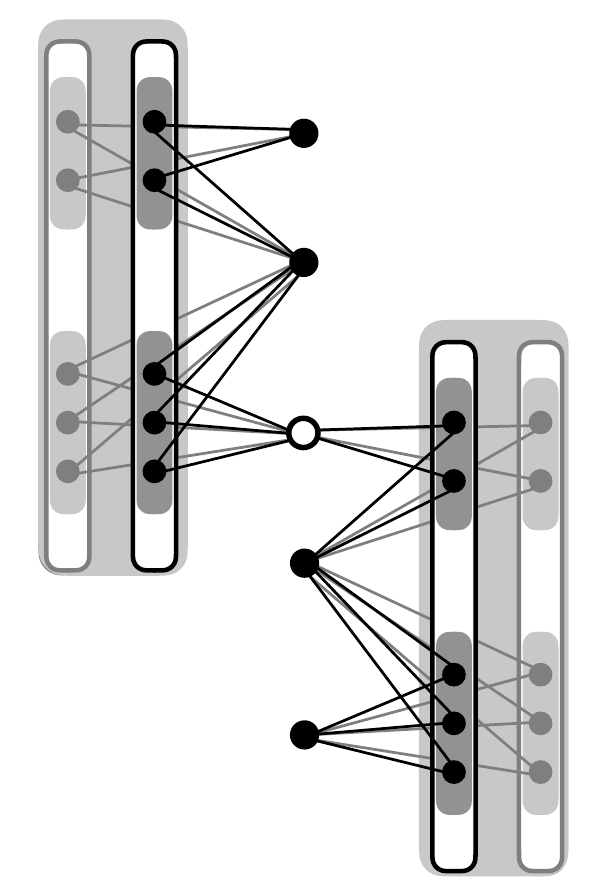}
        \caption{The graph obtained by taking two copies of the graph on the left-hand side and identifying $v_{S_1}$ in the first copy and $v_{S_2}$ in the second copy.}
        \label{fig: 2VClb-glue}
    \end{subfigure}
    \caption{A lower bound graph for $2$-vertex-connectivity. Graph (a) is an accepting instance. Vertices in graph (b) have the same view as in graph (a), while graph (b) should be rejected since the merged (white) vertex is a cut vertex.}
    \label{fig: 2VC-lb}
\end{figure}

Notice that our upper bound for $2$-vertex-connectivity also applies to the anonymous model. This theorem establishes a separation between bounded-degree and unbounded-degree graphs.

%% file: conclusion_v2.tex
\section{Conclusion}\label{sec: conclusion}
In this work, we present a systematic study of the local certification of $k$-edge/vertex-connectivity, establishing new upper and lower bounds. Our results sharpen our understanding of the landscape of what is achievable, but they also leave several intriguing questions open.

\subsection{$k$-edge-connectivity}\label{sec: OQ-kEC}
For $k$-edge-connectivity, we establish a tight bound of $\Theta_k(\log n)$ with respect to $n$. However, the dependence on $k$ is still far from being understood. 
In~\Cref{sec:kec-upper}, we show that for $k \ge 3$, $k$-edge-connectivity can be certified using $O(k \log n)$ bits, while on the lower bound side, we only obtain $\Omega\!\left(\frac{\log n}{k}\right)$. Closing the gap between the current upper and lower bounds, in terms of their dependence on $k$, is an interesting direction for future work.


\subsection{$k$-vertex-connectivity} 
In~\Cref{sec:kvc-upper}, we show that $k$-vertex-connectivity admits certificates of  $\Tilde{O}_k(\sqrt{n})$ bits under the Itai--Zehavi conjecture, while our lower bound remains $\Omega_k(\log n)$. Determining the correct complexity of certifying $k$-vertex-connectivity therefore remains a central open problem arising from this work.


\subsection{Approximate local certification schemes}
\citet{CHPP20} introduced the notion of an \emph{approximation proof-labeling scheme}, which distinguishes graphs satisfying a property from graphs that are far from satisfying it. For example, they showed that certifying graph diameter $\le k$ requires $\Omega\!\left(\frac{n}{k}\right)$ bits, whereas a $2$-approximation requires only $O(\log n)$ bits.

Our results also have implications in this approximate setting. The lower bound constructions in~\Cref{sec: lower} can be adapted to obtain the same lower bound for distinguishing graphs with $(\ge k)$-edge/vertex-connectivity from those with $(\le k-2)$-edge/vertex-connectivity.


Similar to how edge-connectivity is characterized by edge-disjoint spanning trees, a natural analogue for vertex-connectivity is a \emph{connected dominating set partition} (CDS partition), introduced in~\cite{CHGK13}. The best known result~\cite{CHGG+17} shows that every $k$-vertex-connected graph admits a partition into $\Omega\!\left(\frac{k}{\log^2 n}\right)$ connected dominating sets (CDSs). Since each CDS can be certified using $O(\log n)$ bits, this implies that one can certify $\Omega\!\left(\frac{k}{\log^2 n}\right)$-vertex-connectivity using $O\!\left(\frac{k}{\log n}\right)$ bits, assuming the graph is $k$-connected. In other words, this yields an $O(\log^2 n)$-approximation for $k$-vertex-connectivity.

If any exact certification of $k$-vertex-connectivity turns out to require large certificates, is it possible to obtain significantly better approximation ratios using a small number of bits? Additionally, for edge-connectivity, following the discussion in~\Cref{sec: OQ-kEC} on the dependency of $k$ in our certification scheme, it would be interested in investigating whether one can get rid of the dependency on $k$ at the cost of approximation. Understanding this trade-off between accuracy and certificate size for connectivity remains an interesting direction for future work.


\subsection{$2$-edge/vertex-connectivity}
In~\Cref{sec: lower}, our lower bounds for $k$-edge- and $k$-vertex-connectivity are based on indistinguishability from graphs of edge or vertex connectivity at most $2$. The construction has maximum degree $O(k)$, and hence bounded degree for constant $k$. By contrast, \Cref{sec: 2-conn-bdd-deg} shows that on bounded-degree graphs, both $2$-edge-connectivity and $2$-vertex-connectivity admit constant-bit certification schemes. This sharp contrast points to a qualitative difference between certifying $2$-connectivity and certifying $k$-connectivity for $k\ge 3$.

For $2$-edge-connectivity, we extend the constant-bit upper bound from bounded-degree graphs to bounded-expansion graphs. Our approach builds on the pseudoforest certification scheme of~\citet{BFVZ26}, which relies on constant conflict-free colorings. Since such colorings may require $\Omega(\log n)$ bits on general graphs, pushing this approach substantially further appears challenging. For $2$-vertex-connectivity, we prove a superconstant lower bound in the anonymous model; whether a similar lower bound holds in the ID-based setting remains open.

More broadly, the certification complexity of $2$-connectivity is still far from settled. Determining the optimal complexities of certifying $2$-edge-connectivity and $2$-vertex-connectivity, and whether the two problems are inherently different, are natural directions for future work.
